\documentclass[12pt]{amsart}

\usepackage[
text={465pt,580pt},
headheight=9pt,
centering
]{geometry}

\usepackage{amssymb}
\usepackage{amsfonts}
\usepackage{amsthm}
\usepackage{graphicx}
\usepackage{subcaption}
\usepackage{fancyvrb}

\usepackage{color}
\usepackage{url}
\usepackage{xspace}
\usepackage{microtype}
\usepackage{tikz}
\usepackage{enumitem}
\usepackage{hyperref}
\usepackage{pgf,tikz}
\usepackage{tikz-3dplot}
\usepackage{tikz-cd}
\usepackage{float}
\usepackage{graphicx}
\usepackage{cite}
\definecolor{darkblue}{RGB}{0,0,160}
\hypersetup{
colorlinks,%
citecolor=black,%
filecolor=black,%
linkcolor=darkblue,%
urlcolor=darkblue
}

\usepackage[utf8]{inputenc}

\newcommand{\excise}[1]{}%{$\star$\textsc{#1}$\star$}

\newtheorem{thm}{Theorem}[section]
\newtheorem{lemma}[thm]{Lemma}

\newtheorem{cor}[thm]{Corollary}
\newtheorem{prop}[thm]{Proposition}

\theoremstyle{definition}
\newtheorem{example}[thm]{Example}
\newtheorem{remark}[thm]{Remark}
\newtheorem{defn}[thm]{Definition}

\numberwithin{equation}{section}

\newcommand{\ring}[1]{\ensuremath{\mathbb{#1}}}

\renewcommand\>{\rangle}

\newcommand\1{\mathbf{1}}

\newcommand\CC{\ring{C}}

\newcommand\II{\ring{I}}
\newcommand\VV{\ring{V}}
\newcommand\cK{{\mathcal K}}

\newcommand\QQ{\ring{Q}}
\newcommand\RR{\ring{R}}

\newcommand\ZZ{\ring{Z}}

\newcommand\cC{{\mathcal C}}

\newcommand\cL{{\mathcal L}}

\newcommand\cN{{\mathcal N}}
\newcommand\cT{{\mathcal T}}
\newcommand\cS{{\mathcal S}}

\newcommand\cP{{\mathcal P}}

\newcommand\cY{{\mathcal Y}}

\newcommand\cI{{\mathcal I}}

\renewcommand\iff{\Leftrightarrow}

\newcommand{\hadam}{\star}
\renewcommand{\cdot}{}%The way we use it, we do not neet this dot

\renewcommand{\tilde}{\widetilde}

\DeclareMathOperator\sign{sign} % sign
\DeclareMathOperator\im{im} % image
\DeclareMathOperator\rs{rowspace} % rowspace
\DeclareMathOperator\Sat{Sat} % saturation of a lattice
\DeclareMathOperator{\cone}{cone}

\newcommand{\inD}[1][\relax]{\def\argone{#1}\def\temprelax{\relax}
\ifx\argone\temprelax\right.\else\,\middle|#1\right.{}\fi}

\newcommand{\Mathematica}{\texttt{Mathematica}\xspace}
\newcommand{\Maple}{\texttt{Maple}\xspace}
\newcommand{\scip}{\texttt{scip}\xspace}

\allowdisplaybreaks

\begin{document}

\title[Multistationarity in the space of total
concentrations]{Multistationarity in the space of total
  concentrations for systems that admit a monomial parametrization} 

\author{Carsten Conradi}
\address{Hochschule für Technik und Wirtschaft\\ Berlin, Germany}
\email{carsten.conradi@htw-berlin.de}

\author{Alexandru Iosif}
\address{Joint Research Center for Computational Biomedicine\\ Aachen, Germany}
\email{iosif@aices.rwth-aachen.de}

\author{Thomas Kahle}
\address{Otto-von-Guericke Universität\\ Magdeburg, Germany}
\email{thomas.kahle@ovgu.de}

\date{\today}

\makeatletter
  \@namedef{subjclassname@2010}{\textup{2010} Mathematics Subject Classification}
\makeatother

% 13P15   	Solving polynomial systems; resultants
% 13P25   	Applications of commutative algebra (e.g., to
% statistics, control theory, optimization, etc.)
% 80A30   	Chemical kinetics [See also 76V05, 92C45, 92E20]
% 37N25   	Dynamical systems in biology [See mainly 92-XX, but
% also 91-XX]
% 92C42   	Systems biology, networks
% 92B99   Mathematical Biology (in general)


\begin{abstract}
We apply tools from real algebraic geometry to the problem of
multistationarity of chemical reaction networks.  A particular focus
is on the case of reaction networks whose steady states admit a
monomial parametrization.  For such systems we show that in the space
of total concentrations multistationarity is scale invariant: if there
is multistationarity for some value of the total concentrations, then
there is multistationarity on the entire ray containing this value
(possibly for different rate constants) -- and vice versa.  Moreover,
for these networks it is possible to decide about multistationarity
independent of the rate constants by formulating semi-algebraic
conditions that involve only concentration variables. These conditions
can easily be extended to include total concentrations. Hence
quantifier elimination may give new insights into multistationarity
regions in the space of total concentrations.  To demonstrate this, we
show that for the distributive phosphorylation of a protein at two
binding sites multistationarity is only possible if the total
concentration of the substrate is larger than either the total
concentration of the kinase or the total concentration of the
phosphatase. This result is enabled by the chamber
decomposition of the space of total concentrations from polyhedral
geometry. Together with the corresponding sufficiency result of Bihan
et al.\ this yields a characterization of multistationarity up to
lower dimensional regions.
\end{abstract}

\maketitle
\setcounter{tocdepth}{1}
% \tableofcontents

\section{Introduction}
\label{sec:intro}

The dynamics of many biochemical processes can be described by systems
of ordinary differential equations (ODEs).  Already the steady states
of such ODEs contain important information, for example, about the
long term behavior of a process. In particular, in modeling signal
transduction and cell cycle control, one is often interested in the
existence of multiple steady states, multistationarity for
short~\cite{conradi2012multistationarity,degree_paper}.

As measurement data is often noisy and realistic models tend to be
large, parameter values often come with large uncertainties or are not
known at all.  Hence, given an ODE system, one asks whether or 
not there exist parameter values such that the system admits
multistationarity. This is a mathematically challenging problem, even
in the simplest case when all kinetics are of mass-action form and the
ODEs have polynomial right hand sides. In this case one has to
identify those parameter values for which a parametrized family of
polynomials admits at least two positive real solutions.

A variety of necessary conditions for multistationarity in mass-action
networks are known in the literature, for example graph-based
conditions (e.g.~\cite{fein-008,inj-003,inj-002} and the references
therein), conditions based on the determinant of the Jacobian
(e.g.~\cite{inj-010,inj-007} and the references therein) or conditions
based on network concordance (e.g.~\cite{fein-045,fein-060}).

Conditions that are both necessary and sufficient can usually only be
found if the network satisfies additional conditions, for example
encoded in the network
deficiency~\cite{fein-016,fein-017,fein-015,fein-013,fein-014}, in the
stoichiometric matrix~\cite{conradi2012multistationarity}, the steady
state ideal~\cite{millan2012chemical} and~\cite[Section~3.1]{inj-006}
or the Brouwer degree of the polynomial map defining the steady state
ideal~\cite{degree_paper,dickenstein2018multistationarity}.
  The last two references are similar in spirit to the results
  presented here: while we require a monomial parametrization of the
  positive steady states, their results require a rational
  parametrization. Moreover, our results are derived independent of the
  Brouwer degree.

Some of the aforementioned results allow to determine rate
constants where multistationarity occurs. It is, however, currently
not possible to directly infer total concentrations (a different but
equally important set of parameters) based on these results. The
notable exception is \cite{bihan2018lower}, where results similar to
the ones obtained in Section~\ref{subsec:qechambers} are presented: if
the total concentrations satisfy a linear inequality, then there exist
rate constants such that multistationarity is possible. The linear 
inequalities, however, are not arbitrarily imposed by the authors of
\cite{bihan2018lower} but arise from the system itself. It is
therefore currently not possible to decide whether or not
multistationarity is possible for arbitrary polynomial inequalities in
the total concentrations.

As the total concentrations are experimentally more accessible than
the rate constants, conditions directly incorporating total
concentrations are desirable.  Here we initiate the study of such
conditions with a focus on systems whose positive steady states admit
a monomial parametrization (Definition~\ref{def:mono}). These systems
are closely related to systems with toric steady states described
in~\cite{millan2012chemical}, that is, to systems whose steady state
ideal is binomial (i.e.\ the ideal is generated by polynomials with at
most two terms). One way to establish this property is to find a
Gr{\"o}bner basis that is binomial. The results
in~\cite{sadeghimanesh2018groebner} allow for the efficient
computation of such Gr{\"o}bner bases for the enzymatic systems
frequently used in modeling intracellular signaling and control.
Multistationarity conditions are described in~\cite{binomialCore}.
The systems discussed in~\cite{sadeghimanesh2018groebner,binomialCore}
belong to the larger class of MESSI systems~\cite{MESSI}.  However,
while the former systems always admit a monomial parametrization this
need not be the case for the latter.

For systems that admit a monomial parametrization we show that in the
space of total concentrations multistationarity is scale invariant:
from Theorems~\ref{thm:scaling_c} and~\ref{thm:scale} it follows that
if there is multistationarity for some vector $c$ of the total
concentrations (and for some vector of rate constants~$k$), then, for
any $\alpha>0$ there is multistationarity for $\alpha c$, albeit for a
different~$k$.  And vice versa: if for some $c$ there is no $k$ such
that multistationarity is possible, then there is no $k$ such that
multistationarity is possible for $\alpha c$, $\alpha>0$.

In Theorem~\ref{thm:multi_c_space} and Corollary~\ref{cor:semi1} we
formulate semi-algebraic conditions for multistationarity that
use only variables representing concentrations.  Such conditions can
be extended to incorporate constraints on the total concentrations.
Hence, for such systems it is possible to decide about
multistationarity without knowing the rate constants.

There are many biologically meaningful networks that admit a
monomial parametrization, see for example the networks discussed in
\cite{ptm-028}.
We apply our results to one of those, the well-known
sequential distributive phosphorylation of a protein at two binding
sites~\cite{conradi2008multistationarity}
(see~\cite{holstein2013multistationarity} for proteins with an
arbitrary number of phosphorylation sites).  These networks are
arguably among the best studied systems when it comes to
multistationarity: in \cite{sig-005} multistationarity has been shown
numerically, in \cite{fein-012} via sign patterns. This analysis also
allows to study the effect of parameter variations on
multistationarity~\cite{conradi2008multistationarity}.
In~\cite{conradi2014catalytic} conditions on a subset of the rate
constants called catalytic constants have been derived: if the
catalytic constants satisfy this condition, then multistationarity is
possible for some values of the total concentrations.  Here we
describe a similar result for the total concentrations: applying
Corollary~\ref{cor:semi1} we show in Theorem~\ref{theo:no_multi_dd}
and Corollary~\ref{cor:MM-setting} that multistationarity is possible
only if the total concentration of the substrate is at least as
large as either the total concentration of the kinase or the total
concentration of the phosphatase.  A result of \cite{bihan2018lower}
shows the converse up to lower dimensional regions: multistationarity
occurs if the substrate concentration is strictly larger than the
total concentration of kinase or the total concentration of
phosphatase.  Corollary~\ref{cor:MM-setting} summarizes the situation.
The description is complete up to lower dimensional boundary cases.
In particular, multistationarity occurs in the Michaelis--Menten
regime, where the total concentration of the substrate exceeds those
of both enzymes by orders of magnitude.

To arrive at this condition we make use of the chamber decomposition
of the cone of total concentrations: in Theorem~\ref{thm:chambers} we
show that, independent of the number of phosphorylation sites, this
cone consists of five full-dimensional sub-cones called chambers.
These chambers are determined by subsets of linearly independent
columns of a matrix defining the conservation relations.  In
Theorem~\ref{thm:table} we show that for two sites, multistationarity is
only possible in four of these chambers.

  The paper is organized as follows: we close this section with an
  introduction of some basic mathematical notation used throughout
  this paper. In Section~\ref{sec:mass-action-networks} we
  introduce 
  the ODEs that arise from chemical
  reaction networks with mass-action kinetics and formally define
  multistationarity. We also comment on the relationship between
  steady states and rate constants. In Section~\ref{subsec:nonbin} we
  formally define systems that admit a monomial parametrization and
  discuss the consequences concerning steady states and
  multistationarity. In Section~\ref{subsec:qechambers} we apply the
  results of Section~\ref{subsec:nonbin} to a reaction network
  describing the distributive phosphorylation of a protein at two
  binding sites. And in Section~\ref{sec:existence} we comment on
  conditions for the existence of monomial parametrizations. The paper
  closes with a brief discussion of the main findings in
  Section~\ref{sec:disc}.

\subsection{Notation}
\label{sec:notation}
For any $m\times n$ matrix $A$, we write $\im(A) = \{Ax | x\in\RR^n\}$
for the right image and $\rs(A) = \{yA | y\in\RR^m\}$ for the rowspace
(left image).  If $A$ and $B$ are two matrices of the same dimensions,
then $A\hadam B$ denotes their Hadamard product, that is
$(A\hadam B)_{ij} = A_{ij}B_{ij}$.  % Similarly, $\frac{A}{B}$ denotes
% the entry-wise division.
If $x$ is a vector of length $m$ and $A$ is
an $m\times n$ matrix, then we write $x^A$ for the $n$-vector with
entries
\[
(x^A)_j = \prod_{i=1}^mx_i^{A_{ij}}, j=1,\dots,n.
\]
Slightly deviating from the matrix-vector product notation, this
operation is possible independent of whether $x$ is a row or column
vector and should always return the same type of vector.  We also
apply scalar functions to vectors which always means coordinate-wise
application.  Using this, for example, one can check that
\[
\ln x^A = (\ln x)\cdot A \quad \text{ if $x$ is a row vector, }
\]
and
\[
\ln x^A = A^T\cdot (\ln x) \quad \text{ if $x$ is a column vector. }
\]
A vector which has $1$ in every entry is denoted by $\1$. If
$I \subseteq k[x_1,\dots, x_n]$ is an ideal of $n$-variate polynomials
with coefficients in a field $k$, then the variety
$\VV(I) \subseteq k^n$ is the set of all points where all elements of
$I$ simultaneously vanish. See
\cite[Chapter~4]{cox96:_ideal_variet_algor} for basics on
computational algebraic geometry.

\section{Chemical reaction networks}
\label{sec:mass-action-networks}

A \emph{chemical reaction network} is a finite directed graph whose
vertices are labeled by {\em chemical complexes} and whose edges are
labeled by positive parameters called the \emph{rate
  constants} (cf.~(\ref{eq:fexample}) of Example~\ref{ex:1-site}). The
digraph is denoted by $\cN = ([m],E)$, with vertex set $[m]$ and edge
set $E$. 
Each chemical complex $i \in [m]$ has the form $\sum^n_{j=1}
(y_i)_jX_j$ for some $y_i\in\ZZ^n_{\ge0}$, where $X_1,\ldots,X_n$ are
\emph{chemical species}. The vectors $y_i$ are the
\emph{complex-species incidence vectors} and they are gathered as the
columns of the \emph{complex-species incidence matrix} $Y=
(y_1,\ldots,y_m)$. Throughout this article the integers $n$, $m$, and
$r$ denote the number of species, complexes, and reactions,
respectively. Each finite directed graph has an incidence matrix $\cI$
with $\cI_{jl} = - \cI_{il} = 1$ whenever the $\ell^{\text{th}}$ edge 
points from the $i^{\text{th}}$ vertex to the $j^{\text{th}}$ vertex
and $0$ otherwise. A~complex which is the source of a reaction is an
\emph{educt complex} and a complex which is the sink of a reaction is
a \emph{product complex}.  Each complex can be an educt and product
for several reactions. For each reaction network one has a matrix
$\mathcal{Y}$ whose columns are the complex-species incidence vectors
corresponding to 
  the educt complexes for each reaction:
\begin{equation}
\label{eq:def_Ycal}
\mathcal{Y}=(\tilde{y}_1,\ldots,\tilde{y}_r),
\text{ where } \tilde{y}_i=y_k \text{ when reaction $i$ has educt
complex $k$.}
\end{equation} 

  We exemplify the above notation in
  Example~\ref{ex:1-site} below. The system serves as a running example
  in Sections~\ref{sec:mass-action-networks} and~\ref{subsec:nonbin}
  to illustrate our definitions and results.

\begin{example}
  \label{ex:1-site}
  The following reaction network is the \emph{1-site phosphorylation
    network}: 
\begin{equation*}
\label{eq:fexample}
\begin{tikzcd}[row sep=small, column sep=small, every arrow/.append
style={shift left=.75}]
X_1+X_2 \arrow{rr}{k_1} && X_3 \arrow{ll}{k_2} \arrow{rr}{k_3} &&
X_1+X_4 \\
X_4+X_5 \arrow{rr}{k_{4}} && X_6 \arrow{rr}{k_{6}} \arrow{ll}{k_5} && X_2+X_5.
\end{tikzcd}\tag{$\cN_1$}
\end{equation*}
The chemical species are $X_1$, $X_2$, $X_3$, $X_4$, $X_5$, and
$X_6$ and the chemical complexes are $X_1+X_2$, $X_3$,
$X_1+X_4$, $X_4+X_5$, $X_6$, and $X_2+X_5$.  The species
$X_1$ is a catalyst for the phosphorylation of
$X_2$ which goes through an intermediate state
$X_3$ before becoming the phosphorylated~$X_4$.  Similarly
$X_5$ catalyzes the dephosphorylation.  The network has
$6$ reactions, each one labeled by a rate constant $k_1$,
$k_2$, $k_3$, $k_4$, $k_5$ or~$k_6$. The matrix
$\cY$ of this network is
\begin{equation*}
\cY= \left[
\begin{array}{cccccc}
  1&0&0&0&0&0 \\
  1&0&0&0&0&0 \\
  0&1&1&0&0&0 \\
  0&0&0&1&0&0 \\
  0&0&0&1&0&0 \\
  0&0&0&0&1&1  
\end{array}\right].
\end{equation*}
\end{example}

\subsection{Dynamical systems defined by mass-action networks}
\label{sec:dyn-sys-mass-action}

Every chemical reaction network defines a dynamical system of the form
\begin{equation}\label{eq:DiffEq}
\dot{x}=S\nu(k,x),
\end{equation} 
where $S=Y\cI$ is the \emph{stoichiometric matrix} and $\nu(k,x)$ is
the \emph{vector of reaction rates}.  It depends on the vector of
concentrations $x$ and the vector of rate constants~$k$.  % The columns of
% $S$ span the \emph{stoichiometric space} $\cL_{\text{stoi}}$.

In this paper we are concerned with \emph{mass-action networks} for
which the kinetics is of \emph{mass-action form}, i.e.\ the rate of
each reaction is proportional to the product of the concentrations of
its educt complex. Thus, for mass-action networks,
\begin{equation*}
\nu(k,x) = k \hadam \phi(x),
\end{equation*}
where
$\phi(x) = (x^{\tilde{y}_1}, \ldots, x^{\tilde{y}_r})^T =
\left(x^\mathcal{Y}\right)$, and $k=(k_1,\dots,k_r)^T$ is a vector of
parameters.

\begin{example}
The stoichiometric matrix and the monomial vector $\phi(x)$ of
\ref{eq:fexample} are
  \begin{equation*}
    S= \left[
      \begin{array}{rrrrrr}
        -1& 1& 1& 0& 0& 0 \\
        -1& 1& 0& 0& 0& 1 \\
        1&-1&-1& 0& 0& 0 \\
        0& 0& 1&-1& 1& 0 \\
        0& 0& 0&-1& 1& 1 \\
        0& 0& 0& 1&-1&-1  
      \end{array}
    \right] \text{ and }
    \phi(x) = \left(
      \begin{array}{c}
        x_1 x_2 \\ 
        x_3\\
        x_3\\
        x_4 x_5\\
        x_6 \\
        x_6
      \end{array}
    \right).
  \end{equation*}
  The reaction rates are then
  \begin{equation*}
    \nu_1 = k_1 x_1 x_2,\; \nu_2 = k_2 x_3,\; \nu_3 = k_3 x_3,\; \nu_4 =
    k_4 x_4 x_5,\; \nu_5 = k_5 x_6, \text{ and } \nu_6 = k_6 x_6.
  \end{equation*}
  Consequently, the dynamics of \ref{eq:fexample} is given by the
  following system of ODEs:
\begin{equation*}
\begin{array}{ll}
  \dot x_1 = -  k_1x_1x_2 + (k_2 + k_3)x_3,\\
  \dot x_2 = -  k_1x_1x_2 + k_2x_3 + k_6x_6,\\
  \dot x_3 = \hspace{0.3cm} k_1x_1x_2 - (k_2 + k_3)x_3,\\
\end{array} \quad
\begin{array}{ll}
  \dot x_4 = \hspace{0.3cm} k_3x_3 - k_4x_4x_5 + k_5x_6,\\
  \dot x_5 = -  k_4x_4x_5 + (k_5 + k_6)x_6,\\
  \dot x_6 = \hspace{0.3cm} k_4x_4x_5 - (k_5 + k_6)x_6.\\
\end{array}
\end{equation*}
\end{example}

\subsection{Conservation relations and total concentrations}
\label{sec:con-rel}

For many reaction networks there are linear dependencies among
$\dot{x}_1, \ldots, \dot{x}_n$: they are relations of the form
$z\dot{x}=0$, where $z$ is an element of the left kernel of~$S$. If
$z\dot{x}=0$ for $z^T\in\RR^n$ then, by integrating with respect to
time, $zx$ is constant along trajectories.  These constants $zx$ are
the \emph{total concentrations} or \emph{conserved moieties}.  As, by
\eqref{eq:DiffEq}, every $z^T\in\RR^n$ with $zS=0$ yields $z\dot x=0$,
the left kernel of the stoichiometric matrix is called the
\emph{conservation space} $\cL_{\text{cons}}$.  A matrix $Z$ whose
rows are a basis of $\cL_{\text{cons}}$ is a \emph{conservation
matrix}.  
% The \emph{conservation cone} $\cC$ is the set of nonnegative
% vectors in $\cL_{\text{cons}}$.  
In general, every conservation matrix
defines total concentrations via
\begin{equation}
\label{eq:def_tots}
c = Z x.
\end{equation}

  As in our setting the elements of $x$ represent chemical species, we
  are usually only interested in those elements $c$ of
  (\ref{eq:def_tots}) associated to $x\in\RR_{\geq 0}^n$. We use
  the following notation to refer to these
  \begin{equation}
    \label{eq:def_imp_Z}
    \im_+(Z) = \left\{c\in\RR^{n-s} | \exists x\in\RR_{\geq 0}^n \text{
      such that } c = Z x\right \}.
  \end{equation}
  Let $x(0) \in \RR_{>0}^n$ with corresponding $c=Z x(0)$.
  % As describe above, if
  If $x(0) \in \RR^n_{>0}$ is  the initial condition of the  trajectory
  $\{x(t) | t > 0\}$, then, under mass-action kinetics, the above
  discussion implies that $x(t)$ is constrained to the following polyhedron
  associated with $c$:
\begin{equation}
  \label{eq:parametricPoly}
  \cP_{c}= \{ x\in \RR^n_{\ge 0} | Z x = c\}.
\end{equation}
The set $\cP_{c}$ is known as the \emph{invariant polyhedron} with
respect to $x(0)$ \cite{alg-017}, or the
\emph{stoichiometric compatibility class} of
$x(0)$~\cite{fein-016,fein-017}.

  \begin{remark}
  For a given $c\in\RR^{n-s}$ one has $\cP_c\neq\emptyset$ if and only
  $c\in\im_+(Z)$.
  \end{remark}

\begin{example}
The conservation space $\cL_{\text{cons}}$ of \ref{eq:fexample} is
spanned by the rows of the matrix
\begin{displaymath}
Z = \left[
  \begin{array}{cccccc}
    1&0&1&0&0&0 \\
    0&0&0&0&1&1 \\
    0&1&1&1&0&1.
  \end{array}
\right]
\end{displaymath}
Consequently, \ref{eq:fexample} has three linearly independent
conservation relations and three  total concentrations $c_1$, $c_2$
and $c_3$:
  \begin{align*}
    x_1+x_3& =c_1,\\
    x_5+x_6& =c_2,\\
    x_2+x_3+x_4+x_6&=c_3.
  \end{align*}
  The values $c_1$, $c_2$ and $c_3$ can be interpreted as total amount
  of kinase, phosphatase and substrate, respectively. Further
  examples can be found in \cite{CRNT-chapter,shiu2010algebraic}.
\end{example}

\subsection{Steady states}
\label{sec:steady-states}

If $k$ and $x$ are such that 
\begin{equation}
  \label{eq:ss_eq}
  S \nu(k,x) = 0,
\end{equation}
then $x$ is a \emph{steady state}. In the mass-action networks
setting, as $\nu(k,x)$ is a vector of monomials,
equations~\eqref{eq:ss_eq} are algebraic; hence tools from algebraic
geometry are useful in the study of steady states. As $x$ is a vector
of concentrations of chemical species, only nonnegative $x$ are
chemically meaningful. Consequently, when talking about steady states,
we mean nonnegative real solutions of equations~\eqref{eq:ss_eq}.  A
steady state is \emph{positive} if all its coordinates are positive
real numbers.  It is a \emph{boundary steady state} if all coordinates
are nonnegative but it is not positive.  The \emph{steady state ideal}
$I$ is the polynomial ideal generated by the entries of $S\nu(x,k)$.
This ideal can be considered in different polynomial rings. The
parameters $k$ can be part of the indeterminates, i.e.\
$I\subset \RR[x,k]$, or appear as the variables in rational functions
that serve as coefficients.  In the second case $I \subset \RR(k)[x]$.
In both cases the \emph{steady state variety} is the zero locus of the
steady state ideal.

\begin{example}
\label{exa:SteadyStatesNp}
The equations $\dot x_i = 0$ define the steady state ideal of
\ref{eq:fexample}:
\begin{equation*}
\begin{array}{rll}
  I &=
       \langle - k_1x_1x_2 + (k_2 + k_3)x_3,
       - k_1x_1x_2 + k_2x_3 + k_6x_6,k_1x_1x_2 - (k_2 + k_3)x_3,
  \\ & \qquad \quad k_3x_3 - k_4x_4x_5 + k_5x_6,
       - k_4x_4x_5 + (k_5 + k_6)x_6,
       k_4x_4x_5 - (k_5 + k_6)x_6 \rangle\\
   &=
      \langle
      k_1x_1x_2 - (k_2 + k_3)x_3,
      k_3x_3 - k_6x_6,
     k_4x_4x_5 - (k_5 + k_6)x_6 \rangle.
\end{array}
\end{equation*}
The second equality results from elementary simplification and
omitting redundant generators.  While such simplifications are useful
to understand the geometry of steady states, the resulting
polynomials need not have a biochemical interpretation anymore.
\end{example}

When modeling chemical reaction systems, one is often interested in
questions of the form \lq Does there exist a $k$, such that \ldots
?\rq.  The following definition aims to capture such questions by
including both $x$ and $k$ as coordinates.  Following
\cite{alg-041}, we use the word variety, although strictly speaking it
is the positive real part of a variety.
\begin{defn}
\label{def:V+}
The \emph{positive steady state variety} of a reaction network $\cN$
with mass-action kinetics is
\begin{displaymath}
V^+ = \left\{ (k,x)\in\RR^r_{>0} \times \RR^n_{>0} | S\nu(k,x) = 0
\right\}.
\end{displaymath}
\end{defn}

\begin{remark}
  \label{rem:I_V+}
  It would be very interesting to systematically understand the ideal
  $\II(V^+)$ of polynomials that vanish on~$V^+$.  This ideal is
  typically much larger than the steady state ideal.  First, the steady
  state ideal need not contain all functions that vanish on its real
  variety (i.e.\ it need not be a real-radical ideal).  Real-radicals can
  be computed~\cite{neuhaus1998computation,becker1993computation}.  The
  second and more severe problem is that there is no simple method to
  determine~$\II(V^+)$, the ideal of all polynomials that vanish on the
  strictly positive part.  If the steady state equations are binomial
  equations in the $x$ variables (that is if the steady state ideal is
  binomial in $\RR(k)[x]$), then a remedy of sorts is offered at the end
  of Section~\ref{sec:existence}.
\end{remark}

Often it is possible to obtain a parametrization of $V^+$ as shown in
Example~\ref{exa:mono-para-1-site} below. Such parametrizations
simplify the study of multistationarity (which we formally define
after Example~\ref{exa:mono-para-1-site}) and are the topic of
Section~\ref{subsec:nonbin}.

\begin{example}
\label{exa:mono-para-1-site}
According to Example~\ref{exa:SteadyStatesNp}, the steady state ideal
of \ref{eq:fexample} is generated by $3$ polynomials.  Since we are
only interested in positive $x_i$, the equations that describe $V^+$
can be rearranged as
\begin{equation}
\label{eq:binom-para-1-site}
\frac{x_3}{x_6}=\frac{k_6}{k_3},\; \frac{x_1x_2}{x_3}=\frac{k_2+k_3}{k_1},\;
\frac{x_4x_5}{x_6}=\frac{k_5+k_6}{k_4}.
\end{equation} These equations can be solved as
\begin{equation}
\label{eq:mono-para-1-site}
x_3=\frac{k_1}{k_2+k_3}x_1x_2,\;
x_4=\frac{k_1k_3(k_5+k_6)}{(k_2+k_3)k_4k_6}\frac{x_1x_2}{x_5}, \;
x_6=\frac{k_1k_3}{(k_2+k_3)k_6}x_1x_2.
\end{equation}
This shows that the positive steady state variety of \ref{eq:fexample}
can be parametrized by $x_1,x_2$, and $x_5$ together with
$k_1,\ldots,k_6$.  This parametrization uses only products (and
divisions) of the~$x_i$, but no sums.  This \emph{monomial
parametrization} is crucial for the developments of
Section~\ref{subsec:nonbin}.
\end{example}

The following is the central property studied in this paper.
\begin{defn}
  \label{def:multi}
  A network $\cN$ \emph{admits multistationarity} if there are
  $k \in \RR_{>0}^r$ and $a\neq b\in\RR_{>0}^{n}$ such that
  $(k,a) \in V^+, (k,b) \in V^+$, and $a,b\in\cP_{c}$ for $c=Z
    a=Z b$.
\end{defn}
Multistationarity requires the existence of a vector of rate constants
$k$ and an affine subspace $x_0+\im(S)$ that intersects the variety
$\{x | S\nu(k,x) = 0\}$ in at least two distinct positive points.
Often it is useful to have a dual view of this variety: globally, as a
variety in $\RR^r\times\RR^n$, or as a family of varieties in $\RR^n$,
parametrized over~$k$.  The theory of multistationarity is
mathematically interesting because the existential quantifier
``$\exists k\in\RR^r_{>0}$'' can often be eliminated and equivalently
expressed without quantifiers.  Theorems~\ref{thm:multi_c_space}
and~\ref{thm:scaling_c} are instances of this phenomenon.  

\subsection{Steady states and rate constants}

  We revisit the equation (\ref{eq:ss_eq}) and observe that $S
  \nu(k,x)=0$ for $k\in\RR_{>0}^r$ and $x\in\RR_{\geq x}^n$, if and only if
  $\nu(k,x)\in\ker(S) \cap \RR_{\geq 0}^r$ (as $\nu(k,x)$ is nonnegative
  for $k\in\RR_{>0}^r$ and $x\in\RR_{\geq x}^n$). As discussed in -- among
  many other references -- \cite{conradi2012multistationarity},
  $\ker(S) \cap \RR_{\geq 0}^r$ is a pointed polyhedral cone. As such it
  is generated by finitely many generators that are unique up to scalar
  multiplication. Let $E_1,\ldots,E_d$ denote the generators and let
  $E=[E_1, \ldots, E_d]$ be the {\em cone generator matrix}. In
  particular, every generator $E_i$ is nonnegative and every element of
  the cone can be represented by a nonnegative linear combination of the
  generators \cite{lin-004}:
  \begin{displaymath}
    \nu(k,x) \in \ker(S) \cap \RR_{\geq 0}^r \Leftrightarrow
    \nu(k,x)=E\lambda\text{, for some } \lambda\in\RR_{\geq 0}^r. 
  \end{displaymath}
  And $(k,x)\in V^+$, if and only if $\nu(k,x)$ is in the (relative)
  interior of $\ker(S) \cap \RR_{\geq 0}^r$, that is if and only if
  $\nu(k,x) \in \ker(S) \cap \RR_{> 0}^r$. As suggested in
  \cite{conradi2012multistationarity}, the cone $\ker(S) \cap \RR_{>
    0}^r$ can be parametrized with the help of the following set that
  we call {\em coefficient cone}:
  \begin{equation}
    \label{eq:def_LAM}
    \Lambda(E) = \left\{ \lambda\in\RR_{\geq 0}^d | E\lambda >0 \right\}.
  \end{equation}
  See \cite[Remark~4]{conradi2012multistationarity} for more on
  $\Lambda(E)$. For future use we observe:
  \begin{remark}
    \label{rem:LAM_empty}
    $\Lambda(E)=\emptyset$ if and only if $E$ contains a zero row.
  \end{remark}

  \begin{example}
    The network~\ref{eq:fexample} has the following cone generator
    matrix and coefficient cone:
    \begin{displaymath}
      E=\left[
        \begin{array}{ccc}
          1 & 0 & 1 \\
          1 & 0 & 0 \\
          0 & 0 & 1 \\
          0 & 1 & 1 \\
          0 & 1 & 0 \\
          0 & 0 & 1 
        \end{array}
      \right] \text{ and } \Lambda(E) = \RR_{>0}^3.
    \end{displaymath}
    The first connected component of network~\ref{eq:fexample}, can be
    considered as a reaction network in its own right with $x=(x_1$,
    \ldots, $x_4)$ the concentrations of $X_1$, \ldots, $X_4$. This
    network has the following stoichiometric matrix, cone generator
    matrix and vector of reaction rates:
    \begin{displaymath}
      S = \left[
        \begin{array}{rrr}
          -1 & 1 & 1 \\
          -1 & 1 & 0 \\
          1 & -1 & -1 \\
          0 & 0 & 1
        \end{array}
      \right], 
      E=\left[
        \begin{array}{r}
          1 \\ 1 \\ 0
        \end{array}
      \right]
      \text{ and }
      \nu(k,x) = 
      \begin{pmatrix}
        k_1 x_2 x_2 \\ k_2 x_3 \\ k_3 x_3
      \end{pmatrix}
    \end{displaymath}
    In this case the coefficient cone $\Lambda(E)$ is the empty set
    (there is no nonnegative real number $\lambda$ such that
    $E\lambda>0$).  Moreover, $(k,x)$ is a steady state, if and only if
    $\nu(k,x) = E\lambda$.
    From the structure of $E$ if follows that $x_3=0$  for every steady
    state. Hence there are no positive steady states and
    $V^+=\emptyset$ in this case.
  \end{example}

  The following lemma formalizes the observation of the above example.
  \begin{lemma}
    \label{lem:Vplus_E}
    $V^+ \neq \emptyset$ if and only if $E$ does not have a zero
    row. 
  \end{lemma}
  \begin{proof}
  Suppose $V^+\neq\emptyset$, i.e.\ there exist $(k,x)\in V^+$, i.e.\
  $\nu(k,x)=k\hadam \phi(x) \in \ker(S)\cap\RR_{>0}^r$, i.e.\
  $k\hadam\phi(x) = E\lambda$, for some $\lambda\in\Lambda(E)$, and
  thus $\Lambda(E)\neq\emptyset$. Then, by Remark~\ref{rem:LAM_empty},
  $E$ does not have a zero row. 
  Vice versa, suppose $E$ does not have a zero row.  Then
  $\Lambda(E)\neq\emptyset$, again by Remark~\ref{rem:LAM_empty}. Pick
  any $x\in\RR_{>0}^n$ and define $k=\phi(x^{-1})\hadam E \lambda$.
  Then $k\hadam\phi(x) = E\lambda$ and $(k,x)\in V^+$, i.e.\
  $V^+\neq\emptyset$.
  \end{proof}

  The cone generator matrix and the coefficient cone contain important
  information about $V^+$:

\begin{thm}
  \label{thm:k_exists}
  Let $x\in\RR^n_{>0}$ and $k\in\RR_{>0}^r$. If $E$
  does not have any zero row, then
  \begin{displaymath}
    (k,x)\in V^+ \Leftrightarrow \exists \lambda \in \Lambda(E) \text{
      such that } k=\phi\left(x^{-1}\right)\hadam E\lambda.
  \end{displaymath}
\end{thm}
\begin{proof}
  By Definition~\ref{def:V+}, $(k,x) \in V^+$ if and only if
  $S (k \hadam \phi (x)) = 0$ and $k\in\RR^r_{>0}$, $x\in\RR_{>0}^n$.\\
  $\Rightarrow)$ Every element of $\ker(S) \cap \RR_{>0}^r$ is of the
  form $E \lambda$ for some $\lambda \in \Lambda(E)$. Then, if
  $S (k \hadam \phi(x)) = 0$, there is a $\lambda \in \Lambda(E)$ such
  that $k \hadam \phi \left( x \right) = E \lambda$. Hence
  $k = \phi \left( x^{-1} \right) \hadam E \lambda$.\\
  $\Leftarrow)$ If $k = \phi \left( x^{-1} \right) \hadam E \lambda$ for
  some $\lambda \in \Lambda(E)$, then
  $k \hadam \phi \left( x \right) = E \lambda$.  As
  $\forall \lambda \in \Lambda(E)$,
  $E \lambda \in \ker(S) \cap \RR_{>0}^r$, $S (k \hadam \phi (x)) =
  0$. Hence $(k,x) \in V^+$.
\end{proof}

\begin{cor}
  \label{cor:all_x_are_stst}
  If $E$ does not have any zero row then, for every
  $x \in \RR^n_{>0}$, there is a $k \in \RR^r_{>0}$ such that
  $(k,x)\in V^+$.
\end{cor}
\begin{proof}
    Suppose $E$ does not have a zero row. Then $V^+\neq\emptyset$ by
    Lemma~\ref{lem:Vplus_E}. Pick any $\lambda\in\Lambda(E)$ and
    define $k=\phi(x^{-1})\hadam E \lambda$. Then $k\hadam \phi(x) =
    E\lambda$ which is equivalent to
    $\nu(k,x)\in\ker(S)\cap\RR_{>0}^r$.  Hence $(k,x)\in V^+$.
\end{proof}

  \begin{remark}
    Theorem~\ref{thm:k_exists} shows in particular that under the (very
    mild) assumption that $E$ does not have any zero row, for every
    positive $x$, one can find positive $k$ such that $(k,x)\in V^+$.  If
    $E$ does have zero rows, then there are no positive steady states and
    $V^+$ is empty. Hence from here on we only consider reaction
    networks $\cN$ where the cone generator matrix $E$ does not have any
    zero row.
  \end{remark}

\section{Monomial parametrizations of positive steady states}
\label{subsec:nonbin}

In this section we consider a mass-action network $\cN$ on $n$ species
and $r$ reactions, with at least one conservation relation. 
  In Section~\ref{subsec:qechambers} we use the results of this
  section to deduce conditions for multistationarity in the space of
  total concentrations for a network describing the distributive
  phosphorylation of a protein.

Let $S$ and $Z$ denote the stoichiometric and a conservation matrix of
$\cN$ respectively.  We study the consequences of the existence of
monomial parametrizations for the positive steady state variety
of~$\cN$.  Following \cite{inj-006}, the positive steady state variety
admits a monomial parametrization if suitable Laurent monomials in the
concentrations can be expressed in terms of the reaction rates
(cf.\cite[Section~3.2]{inj-006}.  Such systems can be diagonalized
using monomial transformations.  The following definition captures
what was observed in Example~\ref{exa:mono-para-1-site}.

\begin{defn}
  \label{def:mono}
  The positive steady state variety $V^+$ \emph{admits a monomial
    parametrization} if there is a matrix $M\in\ZZ^{n\times d}$ of rank
  $p<n$ and a rational function $\gamma$ in the variables
  $k_1,\dots,k_r$ with values in $\RR^d$ such that, for all
  $(k,x) \in \RR^r_{>0} \times \RR^n_{>0}$,
  \begin{displaymath}
    (k,x) \in V^+ \Leftrightarrow \text{$\gamma(k)$ is defined and }
    x^{M} = \gamma(k).
  \end{displaymath}
\end{defn}

In Definition~\ref{def:mono}, the matrix $M$ is understood as part of
saying \emph{admits a monomial parametrization}.  In the following, if
$V^+$ admits a monomial parametrization and a matrix $M$ appears, then
it is the matrix in that definition.

The existence of a monomial parametrization implies that all positive
steady states can be recovered from monomial transformations of one
positive steady state.  In algebraic geometry, a variety which
equals the closure of an algebraic torus acting on the variety is
known as a toric variety.  Affine toric varieties are cut out by
binomial equations such as those in Definition~\ref{def:mono}.

Equation~\eqref{eq:binom-para-1-site} of
Example~\ref{exa:mono-para-1-site} shows that the
network~\ref{eq:fexample} admits a monomial para\-metri\-zation according
to Definition~\ref{def:mono}.  By introducing two matrices $M^+$ and
$M^-$ with nonnegative entries, of appropriate dimension, such that
\begin{equation}
\label{eq:defM_plus-minus}
M = M^+ - M^-,
\end{equation}
and extracting numerators and denominators of the rational function
$\gamma(k)$ as follows
\begin{equation}
\label{eq:def_g_plus_minus}
\gamma^\pm(k) = (\gamma^\pm_i(k))_i, \text{ where }
\gamma_i(k) = \frac{\gamma_i^-(k)}{\gamma_i^+(k)},
\end{equation}
we can write the system of Definition~\ref{def:mono} as a binomial
system:
\begin{equation}
\label{eq:def_bino_plus_minus}
\gamma^+(k) \hadam x^{M^+} - \gamma^-(k) \hadam x ^{M^-} = 0.
\end{equation}
\begin{remark}
\label{rem:ES96}
In \eqref{eq:def_g_plus_minus} the coefficients $\gamma^\pm(k)$
usually have many terms, hence \eqref{eq:def_g_plus_minus} is binomial
only in~$x$.  As a consequence of \cite[Theorem~2.1]{ES96}, the ideal
$\<x^M - \gamma(k)\> \subset \RR(k)[x^\pm]$ is a complete
intersection.  This means that there exists a generating set
of $\<x^M - \gamma(k)\>$ in which $M$ has full rank.  In the following
we assume that $M$ is of full rank.
\end{remark}
\begin{remark}
Our Definition~\ref{def:mono} is equivalent to asking that the ideal
$\II(V^+)$, considered in the ring $\RR(k)[x]$, is generated by the
binomials \eqref{eq:def_bino_plus_minus}.  As there may well be several
generating sets of binomials, neither the coefficients
$\gamma^{\pm}(k)$ nor the matrices $M^\pm$ need be unique.  In
\cite{binomialCore} a similar situation is considered: the ideal
defined by the polynomials $S \nu(k,x)=0$ is generated by binomials in
the ring $\RR(k)[x]$. Our definition is slightly more general, as it
might happen that even though $S \nu(k,x)=0$ is not generated by
binomials, the ideal $\II(V^+)$ is.
\end{remark}
Given the binomials~\eqref{eq:def_bino_plus_minus}, we can now define
the positive values of $k$ where the vector $\gamma(k)$ of
Definition~\ref{def:mono} is defined: the system
\eqref{eq:def_bino_plus_minus} can only be satisfied by positive $k$
and $x$ if the coefficients $\gamma^\pm(k)$ are nonzero and of the
same sign, that is if $k$ is contained in the semi-algebraic set
\begin{equation}
\label{eq:L_gamma}
\cK^+_\gamma := \left\{ k\in\RR_{>0}^r | \gamma_i^+(k)\; \cdotp \gamma_i^-(k) >
  0, i = 1, \ldots, p \right\}.
\end{equation}
In particular, if $k\notin \cK^+_\gamma$, then there does not exist a
vector $x\in\RR_{>0}^n$ such that $(k,x)\in V^+$.

The next few lemmata make the monomial parametrization explicit in our
setting. 
\begin{lemma}
\label{lem:fibre}
If $V^+\!$ admits a monomial parametrization and, for $q<n$,
$A\in\QQ^{q\times n}$ is any matrix of maximal rank $q$ such that
$A M = 0$, then:
\begin{enumerate}[label={(\roman*)}]
\item\label{it:l1a}
$(k,x)\in V^+ \Leftrightarrow (k,x\hadam \xi^A) \in V^+$,
$\forall \xi \in \RR_{>0}^q$,
\item\label{it:l1b} $(k,x)\in V ^+ \Leftrightarrow (k,x\hadam
(e^{\kappa})^A) \in V^+$,
$\forall \kappa \in \RR^q$.
\end{enumerate}
\end{lemma}

\begin{proof} 
As $V^+$ admits a monomial parametrization, the left hand side of
\ref{it:l1a} is equivalent to $x^{M} = \gamma(k)$ and the right hand
side is equivalent to $(x\hadam\xi^A)^{M}=\gamma(k)$.  As $AM=0$,
these are equivalent:
$\left(x\hadam \xi^A \right)^{M}= x^{M}$.  Item
\ref{it:l1b} follows from \ref{it:l1a} by replacing $\xi$
with~$e^\kappa$.
\end{proof}

By Lemma~\ref{lem:fibre},
given a pair $(k,x) \in V^+$, one obtains all $\tilde x$ with
$(k,\tilde x) \in V^+$ from $x$ with the help of the left kernel of
$M$. In the following lemma we show that by choosing a special basis
of the left kernel of $M$, one can make the connection between $x$ and
$k$ in the solution of $S \nu(k,x)=0$ explicit.

\begin{lemma}
\label{lem:psi-rep}
Assume $V^+$ admits a monomial parametrization. Then there are
\begin{itemize}
\item a matrix $A\in\QQ^{(n-p)\times n}$ of rank $n-p$ such that
$AM = 0$,
\item a function $\psi:\mathcal{K}_\gamma^+ \to \RR^n$, and
\item an exponent $\eta\in\ZZ_{>0}$,
\end{itemize}
such that $\psi^\eta$ is a rational function and
\begin{displaymath}
(k,x) \in V^+ \Leftrightarrow k\in \mathcal{K}^+_\gamma \text{ and } \; \exists\,
\xi \in\RR_{>0}^{n-p} \text{ such that } x = \psi(k)\hadam \xi^A.
\end{displaymath}
\end{lemma}

\begin{proof}
  As in Remark~\ref{rem:ES96}, consider the ideal
  $\<x^M - \gamma(k)\> \subset \RR(k)[x^\pm]$.
  By\cite[Theorem~2.1]{ES96}, this ideal is a complete intersection and
  we can find a generating set in which $M$ has full rank and format
  $n\times p$ for a suitable~$\gamma$. 
    In the following we assume that $M$ is ordered such that the first
    $p$ rows are linearly independent. (Note that this can always be
    achieved by a suitable reordering of the variables $x$.)
  Then there is an invertible matrix
  $U\in\QQ^{p\times p}$ such that
  \begin{displaymath}
    MU =
    \begin{bmatrix}
      I_p \\
      -W 
    \end{bmatrix},
  \end{displaymath}
  where $W$ is of format $(n-p)\times p$.
    Let
    \begin{displaymath}
      A = [W | I_{n-p}] \text{ and } \psi(k) = \gamma(k)^{[U | 0_{p\times
          n-p}]}.
    \end{displaymath}
    We now argue that $A$ and $\psi(k)$ have the properties stated above:
    first, $AM = 0$ since $AMU = 0$ and $U$ is invertible. Second,
    only if $k\in\cK^+_\gamma$, there exists $x\in\RR_{>0}^n$ such that
    $(k,x)\in V^+$. Hence we only need to consider $k\in\cK^+_\gamma$.
    As $\gamma_i^\pm(k)\neq 0$ for $k\in\cK^+_\gamma$ by
    \eqref{eq:L_gamma} and, as powers of the entries $\psi_i(k)$ are products of
    $\gamma_i^\pm$, the function $\psi(k)$ is well defined on $\cK^+_\gamma$.
    And third, as $\gamma(k)$ is rational, the coordinate-wise power
    $\psi^{\eta}(k)$ is rational when $\eta$ is the least common
    multiple of the denominators in~$U$.

  Next we turn to the equivalence:
  according to Definition~\ref{def:mono},
  \begin{displaymath}
    (k,x) \in V^+ \Leftrightarrow x^{M} = \gamma(k) \text{ and } k\in\cK_\gamma^+.
  \end{displaymath}
  In the following calculations we take logarithms on both sides of the
  equation above.  This is well defined, if we require
    $k\in\cK^+_\gamma$ which implies $\gamma(k)>0$.
  Taking logarithms we get
  \begin{displaymath}
    M^T\cdot(\ln x) = \ln \gamma(k) \Leftrightarrow U^TM^T(\ln x) =
    U^T(\ln \gamma(k)).
  \end{displaymath}
  Decompose $x$ into $x' = (x_1,\dots,x_p)^T$ and
  $\xi = (x_{p+1},\dots,x_n)^T$. 
    As $U^TM = [I_p | -W^T]$ the above equivalence is
  \begin{align*}
    (k,x) \in V^+ 
    \Leftrightarrow
    \ln {x}' - W^T\cdot(\ln  \xi) = U^T\cdot(\ln \gamma(k))
    \Leftrightarrow     
    {x}' = \gamma(k)^{U} \hadam \xi^W .
  \end{align*}
    Using $x=(x',\; \xi)$ and the above matrix $A$ together with the
    vector $\psi(k)$ we obtain the final equivalence
    $(k,x) \in V^+ 
      \Leftrightarrow x = \psi(k) \hadam \xi^A.
      $
\end{proof}

  \begin{remark}
    \begin{enumerate}[label={(\roman*)}]
    \item For fixed $k\in\mathcal{K}^+$, the matrix $A$ in
      Lemma~\ref{lem:psi-rep} captures all information about the
      parametrization.  We call $A$ the \emph{exponent matrix} of the
      parametrization. 
    \item Choosing $\xi_i=1$ for all $i$, one obtains $(k,\psi(k))\in V^+$,
      i.e.\ the vector $\psi(k)$ is a (positive) solution of
      the equation $S\nu(k,x)=0$ for a given vector~$k$.
    \item In the proof of Lemma~\ref{lem:psi-rep}, the coordinates
      $\xi$ are elements of $x$, i.e.\ the $\xi_i$ correspond to
      variables of the system and thus have a biological
      meaning. Moreover, usually there are several orderings of
      variables one can choose from when constructing the matrix $A$
      in the proof of Lemma~\ref{lem:psi-rep}. One strategy would be
      to choose an ordering that yields $\xi_i$ that correspond to
      measurable species.
    \end{enumerate}
  \end{remark}

The following example is an illustration of
Lemma~\ref{lem:psi-rep} and the steps taken in its proof.

\begin{example}
    Going back to Example~\ref{exa:mono-para-1-site}, 
  equations~\eqref{eq:binom-para-1-site} can be expressed as
\[
x^M = \gamma(k),
\]
where
\[
M = \left[
\begin{array}{rrrrrr}
  0 & 0 & 1 & 0 & 0 &-1 \\
  1 & 1 &-1 & 0 & 0 & 0 \\
  0 & 0 & 0 & 1 & 1 &-1 
\end{array} \right]^T
\text{ and } \
\gamma(k) =
\left( 
\frac{k_6}{k_3},
\frac{k_2+k_3}{k_1},
\frac{k_5+k_6}{k_4}
\right)^T.
\]
As numerators and denominators of $\gamma(k)$ are sums of positive
monomials, one has $\cK^+_\gamma=\RR_{>0}^6$, that is, the monomial
parametrization is valid for all positive $k$.  For example, for the
matrix
\begin{displaymath}
  U = \left[
    \begin{array}{rrr}
      0 &-1 &-1 \\
      -1&-1 &-1 \\
      0 & 0 & 1
    \end{array} \right],
\text{ one obtains }
  M U = \left[
    \begin{array}{rrr}
      -1&-1&-1 \\
      -1&-1&-1 \\
      1&0&0 \\
      0&0&1 \\
      0&0&1 \\
      0&1&0 
    \end{array}
  \right],
\end{displaymath}
which in the ordering $(x_3,x_6,x_4,x_1,x_2,x_5)^T$ is equivalent to
\begin{displaymath}
  \left[
    \begin{array}{r}
      I_3 \\
      -W
    \end{array}
  \right] \text{ with }
  W = \left[
    \begin{array}{rrr}
      1 & 1 & 1 \\
      1 & 1 & 1 \\
      0 & 0 &-1
   \end{array} \right].
\end{displaymath}
  As in the proof of Lemma~\ref{lem:psi-rep}, we obtain
  \begin{displaymath}
    A = [W | I_3] = \left[
      \begin{array}{rrrrrr}
      1 & 1 & 1 & 1 & 0 & 0 \\
      1 & 1 & 1 & 0 & 1 & 0 \\
      0 & 0 &-1 & 0 & 0 & 1
      \end{array}
    \right].
  \end{displaymath}
  For $\psi(k)$ we obtain 
  \begin{displaymath}
    \psi(k) = \gamma(k)^{[U\, | \, 0_{3\times 3}]} = 
    \left(
      \frac{k_1}{k_2+k_3},
      \frac{k_3}{k_6} \frac{k_1}{k_2+k_3},
      \frac{k_3}{k_6} \frac{k_1}{k_2+k_3}  \frac{k_5+k_6}{k_4}, 
      1,1,1
    \right)^T.
  \end{displaymath}
  In this case $\psi(k)$ is already a rational function, as the the matrix
  $U$ contains only integer entries and the least common multiple of
  the denominators in $U$ therefore is $\eta=1$.  If $U$ was
  not an integer matrix but contained rational entries, then $\psi(k)$
  would not be a rational function as it would contain entries with
  rational exponents. In this case only $\psi(k)^\eta$ would be
  rational as taking entries of $\psi$ with rational exponents to
  the power $\eta$ yields integer exponents.
Then, with $\xi=(\xi_1,\xi_2,\xi_3)^T$,
\begin{multline*}
  (x_3,x_6,x_4,x_1,x_2,x_5)^T % = \psi(k) \hadam \xi^{[W | I_3]} 
  = \psi(k) \hadam \xi^A \\
  = 
  \left(
    \frac{k_1}{k_2+k_3}\xi_1\xi_2,\;
    \frac{k_1k_3}{(k_2+k_3)k_6}\xi_1\xi_2,\;
    \frac{k_1k_3(k_5+k_6)}{(k_2+k_3)k_4k_6}\frac{\xi_1\xi_2}{\xi_3},\;
    \xi_1,\;
    \xi_2,\;
    \xi_3
  \right)^T.
\end{multline*}
\end{example}

  As explained above, if $V^+$ admits a monomial parametrization and
  if $(k,x)\in V^+$, then there are infinitely many $\tilde x$ with
  $(k,\tilde x)\in V^+$. The following result describes the connection
  between an arbitrary pair $a$ and $b$ of these.

\begin{lemma}
\label{lem:mono-connect}
If $V^+$ admits a monomial parametrization with exponent matrix
$A\in\QQ^{(n-p)\times n}$ and $k\in\cK_\gamma^+$ and
$a \neq b \in\RR_{>0}^{n}$ are such that $(k,a)\in V^+$ and
$(k,b) \in V^+$, then
\begin{enumerate}[label={(\roman*)}]
\item\label{it:l2a} $\exists\, \xi\neq \1 \in\RR_{>0}^{n-p}$ such that
  $b=a\hadam \xi^A$, 
\item\label{it:l2b} $\exists\, 0\neq \mu\in\rs(A)$ such that
$b=a\hadam e^\mu$.
\end{enumerate}
\end{lemma}
\begin{proof} \ref{it:l2a} It follows from Lemma~\ref{lem:psi-rep}
that there are $\xi_1, \xi_2\in\RR_{>0}^{n-p}$ such that
$a=\psi(k)\hadam \xi_1^A$ and $b=\psi(k)\hadam \xi_2^A$. Then,
$\psi(k) = a\hadam \xi_1^{-A}$ and
$b=a\hadam \xi_1^{-A}\hadam \xi_2^A = a \hadam \xi^A$ with
$\xi = \frac{\xi_2}{\xi_1}$. Item \ref{it:l2b} follows from
\ref{it:l2a} by replacing $\xi^A$ with $(e^{\ln(\xi)})^A$.
\end{proof}

This final corollary summarizes the development so far.
\begin{cor}\label{c:findk}
If $V^+$ admits a monomial parametrization with exponent matrix
$A\in\QQ^{(n-p)\times n}$, then for every positive $x\in\RR^n_{>0}$
there exists a vector 
  $k\in\mathcal{K}^+_\gamma$
such that the following equivalent conditions hold:
\begin{enumerate}[label={(\roman*)}]
\item\label{it:findka} $(k,x) \in V^+$,
\item\label{it:findkb} $x^{M} = \gamma(k)$, % and $k\in\cK^+_\gamma$,
\item\label{it:findkc} $\exists \xi\in\RR^{n-p}_{>0}$ such that
${x}=\psi(k)\hadam \xi^A$. % and $k \in \mathcal{K}^+_\gamma$.
\end{enumerate}
\end{cor}
  \begin{proof}
    This is Lemma~\ref{lem:psi-rep} together with
    Theorem~\ref{thm:k_exists} and Corollary~\ref{cor:all_x_are_stst}.
  \end{proof}

\subsection{Multistationarity}
This section collects results concerning multistationarity under the
assumption that $V^+$ admits a monomial parametrization. Some
conditions involve sign patterns similar to
\cite{conradi2012multistationarity} and~\cite{inj-006}.  For a scalar
$u$ we use $\sign(u)$ to denote its sign, for a vector $v\in\RR^n$ we
use $\sign(v)=(\sign(v_1), \ldots, \sign(v_n))$ to denote its sign
pattern.  Theorem~\ref{thm:signs_multi} appeared in a different
formulation in \cite{inj-006}.
  In this subsection we frequently refer to~$Z$, the conservation
  matrix of a reaction network and the set $\im_+(Z)$. Our first
  result exploits a monomial parametrization of $V^+$ to formulate
  conditions for multistationarity that are independent of the rate
  constants. 

\begin{lemma}
  \label{lem:rep_multi}
  If $V^+$ admits a monomial parametrization with exponent matrix $A
  \in \QQ^{(n-p)\times n}$, then the following are equivalent:
\begin{enumerate}[label={(\roman*)}]
\item\label{it:l3a} $\cN$ admits multistationarity,
\item\label{it:l3b} $\exists$ $x\in\RR_{>0}^n$ and
$\xi\in\RR_{>0}^{n-p}\setminus\{\1\}$, such that
$Z (x- x\hadam \xi^A) = 0$,
\item\label{it:l3c}$\exists$ $x\in\RR_{>0}^n$ and
$\kappa \in \RR^{n-p} \setminus \{0\}$, such that
$Z (x- x\hadam (e^{\kappa})^A) = 0$.
\end{enumerate}
\end{lemma}
\begin{proof}
Items \ref{it:l3b} and \ref{it:l3c} are equivalent as for any
$\xi \in \RR^{n-p}_{>0}$ there is a $\kappa \in \RR^{n-p} $ such that
$\xi = e^\kappa$. Now assume \ref{it:l3b} holds for some $x$ and
$\xi$.  We prove that \ref{it:l3a} holds. By Lemma~\ref{lem:psi-rep},
there exists a $k \in \cK^+_\gamma$ such that $(k,x) \in V^+$ and
   by Lemma~\ref{lem:fibre} $(k,x\hadam\xi^A)\in V^+$ as well. Since
   $Zx=Z(\xi^A\hadam x)=c$ by assumption one has $x$, $\xi^A\hadam x
   \in \cP_c$, that is, 
$\mathcal{N}$ admits
multistationarity.  When \ref{it:l3a} holds, we have $x\neq x'$ and
$k$ such that $Z(x-x') = 0$, and $(k,x) \in V^+$ and $(k,x')\in V^+$.
Now Lemma~\ref{lem:mono-connect} implies $x' = x\hadam \xi^A$ and thus
\ref{it:l3b}.
\end{proof}

  As discussed in \cite{millan2012chemical} and \cite{inj-006} for
  systems with toric steady states multistationarity can be
  established by analysis of sign
  patterns. Theorems~\ref{thm:compute_abk} -- \ref{thm:signs_multi}
  below translate this to our setting. Theorems~\ref{thm:compute_abk}
  and \ref{thm:patterns_exist} show that the existence of a pair of
  nontrivial vectors $\mu\in\rs(A)$ and $z\in\im(S)$ with
  $\sign(\mu)=\sign(z)$ is both necessary and  sufficient for
  multistationarity. Theorem~\ref{thm:compute_abk} is constructive in
  the sense that given such a pair $\mu$, $z$ one can construct rate
  constants $k$ and a corresponding pair of steady states $a$ and $b$.

\begin{thm}
\label{thm:compute_abk}
If $V^+$ admits a monomial parametrization with exponent
matrix~$A\in\QQ^{(n-p)\times n}$ and there are $\mu \in\rs(A)$ and
$z \in\im(S)$ such that $\sign(\mu) = \sign(z)$, then $\cN$ admits
multistationarity.  Specifically, for arbitrary
$\bar a_i \in \RR_{>0}$, $i \in [n]$, let $a\in\RR_{>0}^n$ denote the
vector with entries
\begin{subequations}
\begin{align}
  \label{eq:def_a}
  a_i & =
	\begin{cases}
        \frac{z_i}{e^{\mu_i}-1} & \text{ if }z_i \neq 0, \\
        \bar a_i & \text{ else,}
	\end{cases}
  \intertext{and let}
  \label{eq:def_b}
     b & = a \hadam e^\mu.
   \intertext{Then, for any $\lambda\in\Lambda(E)$, setting}
   \label{eq:def_k}
	 k &= \phi(a^{-1})\hadam E\lambda,
\end{align}
\end{subequations}
$\cN$ admits multistationarity as
\begin{displaymath}
(k,a) \in V^+, \ (k,b) \in V^+, \text{ and } \ (b-a)\in\im(S).
\end{displaymath}
\end{thm}
\begin{proof}
The vector $b$ is positive whenever $a$ is positive, and the vector
$a$ is positive, whenever $\sign(\mu)=\sign(z)$.  By definition,
$(b-a) = z \in \im(S)$.  Then Theorem~\ref{thm:k_exists} shows
$(k,a)\in V^+$ and Lemmata \ref{lem:fibre} and~\ref{lem:psi-rep} also
show $(k,b)\in V^+$.
\end{proof}

\begin{thm}
\label{thm:patterns_exist}
Assume $V^+$ admits a monomial parametrization with exponent matrix
$A\in\QQ^{(n-p)\times n}$ and let $k\in\cK^+_\gamma$ and
$a,b\in\RR_{>0}^n$, $a\neq b$, be such that $(k,a) \in V^+$,
$(k,b)\in V^+$, and $(b-a)\in\im(S)$.  Let $z=b-a$ and
$\mu=\ln b - \ln a$. Then
\begin{enumerate}[label={(\roman*)}]
\item \label{it:paternsa} $z\in\im(S)$, $\mu\in\rs(A)$,
$\sign(z)=\sign(\mu)$,
\item \label{it:paternsb} $k$, $a$, and $b$ together with $z$ and $\mu$
satisfy~\eqref{eq:def_a} -- \eqref{eq:def_k}.
\end{enumerate}
\end{thm}
\begin{proof} For item~\ref{it:paternsa}, $z \in \im(S)$ by
assumption. As $V^+$ admits a monomial parametrization, by
Lemma~\ref{lem:psi-rep}, there are $\kappa_1$ and
$\kappa_2\in\RR^{n-p}$ such that $a=\psi(k)\hadam (e^{\kappa_1})^A$
and $b=\psi(k)\hadam (e^{\kappa_2})^A$.  Hence
$\mu = (\kappa_2-\kappa_1)A$ and, consequently, $\mu\in\rs(A)$. By
construction $b=e^\mu\hadam a$, and thus $z=(e^\mu-\1)\hadam a$.
As $a$ is positive, $\sign(e^\mu-\1)=\sign(z)$ must hold. As
$\sign(e^\mu-\1)=\sign(\mu)$, $\sign(\mu)=\sign(z)$. For
item~\ref{it:paternsb}, \eqref{eq:def_b} holds by construction and
\eqref{eq:def_a} follows from the equation $z=(e^\mu-\1)\hadam a$.
Now, $(k,a)\in V^+$ implies that $k\hadam \phi(a) = E\lambda$ for some
$\lambda\in\Lambda(E)$ by Theorem~\ref{thm:k_exists}; hence
\eqref{eq:def_k} also holds. 
\end{proof}

  The following Theorem~\ref{thm:signs_multi} rests on the set of all
  sign patterns associated to a linear subspace: let
  $\mathcal{U}\subseteq\RR^n$ be a linear subspace, then
  $\sign(\mathcal{U})$ is the set of all sign patterns of all its
  elements:
  \begin{equation}
    \label{eq:def_sign_subspace}
    \sign(\mathcal{U}) = \left\{ \delta\in\{-1,0,1\}^n | \exists
      u\in\mathcal{U} \text{ with } \sign(u) = \delta \right\}
  \end{equation}
  \begin{example}
    Let $\mathcal{U} = \im\left(
      \begin{bmatrix}
        -5\\ \phantom{-}2
      \end{bmatrix}
    \right)$ and observe that for every vector $u\in\mathcal{U}$ one
    has either
    \begin{displaymath}
      \sign(u) = 
      \begin{pmatrix}
        -1 \\ \phantom{-}1
      \end{pmatrix} \text{ or } \sign(u) =
      \begin{pmatrix}
        0 \\ 0
      \end{pmatrix} \text{ or } \sign(u) =
      \begin{pmatrix}
        \phantom{-}1 \\ -1
      \end{pmatrix}.
    \end{displaymath}
    Consequently
    \begin{displaymath}
      \sign(\mathcal{U}) =   
      \left\{    
        \begin{pmatrix}
          -1 \\ \phantom{-}1
        \end{pmatrix},
        \begin{pmatrix}
          0 \\ 0
        \end{pmatrix},
        \begin{pmatrix}
          \phantom{-}1 \\ -1
        \end{pmatrix}\right\}.
    \end{displaymath}
  \end{example}

Theorem~\ref{thm:signs_multi} below is similar to \cite[Proposition~3.9 and
Corollary~3.11]{inj-006}. All of them employ analysis of the
aforementioned sign patterns to decide the existence of two positive
real solutions $a$ and $b$ to the parametrized family of
polynomials (\ref{eq:ss_eq}) such that both are elements of the affine
space $\{x| Zx = Za = Zb\}$. For a detailed discussion on how
to verify sign conditions, see \cite[Section~4]{inj-006}.

\begin{thm}
\label{thm:signs_multi}
If $V^+$ admits a monomial parametrization with exponent matrix~$A$,
then there are $k\in\cK^+_\gamma$ and $a\neq b \in \RR_{>0}^n$ such
that $(k,a) \in V^+$, $(k,b)\in V^+$, and $Z (b-a)=0$ if and only if
\begin{equation}
\label{eq:sign_condi}
\sign(\rs(A)) \cap \sign(\im(S)) \neq \{ 0 \}. 
\end{equation}
\end{thm}
\begin{proof}
This is the combination of Theorems~\ref{thm:compute_abk}
and~\ref{thm:patterns_exist}.
\end{proof}

\subsection{Multistationarity in the space of total concentrations}
\label{sec:totConc}

In this section we study multistationarity in the space of total
concentrations under the assumption that $V^+$ admits a monomial
parametrization. The first result translates some of the results of
the previous section into the space of total concentrations.

\begin{thm}
  \label{thm:multi_c_space}
  Assume $V^+$ admits a monomial parametrization with exponent matrix
  $A\in\QQ^{(n-p)\times n}$. 
    Let $c$ be an element of $\im_+(Z)$,
  then the following are equivalent:
\begin{enumerate}[label={(\roman*)}]
\item\label{it:t1a} 
  $ \exists k\in\cK_\gamma^+$ and $a \ne b\in \RR_{>0}^{n}$ such that
  $(k,a), (k,b)\in V^+$, 
    and $c=Z a = Z b$, that is $a$, $b \in \cP_c$,
\item\label{it:t1b} 
  $\exists k\in\cK_\gamma^+$ 
  such that $Z (\psi(k)\hadam \xi^A) = c$ has at least two solutions $\xi_1 \ne
  \xi_2\in\RR^{n-p}_{>0}$, 
\item\label{it:t1c} 
  $\exists a\in\RR_{>0}^n$ and $\xi\neq \1 \in\RR_{>0}^{n-p}$, such that
  $Z (a\hadam \xi^A -a) = 0$
    and $c=Z a = Z b$.
\end{enumerate}
\end{thm}
\begin{proof}
  \ref{it:t1a}$\Rightarrow$\ref{it:t1b}: 
    Let $k\in\cK_\gamma^+$ and $a\neq b \in \RR_{>0}^{n}$ as in
    \ref{it:t1a}. By Lemma~\ref{lem:psi-rep},
    there are $\xi_1, \xi_2\in\RR_{>0}^{n-p}$ such that
    $a=\psi(k)\hadam \xi_1^A$ and $b=\psi(k)\hadam \xi_2^A$ and $a\neq
    b$ implies $\xi_1\neq\xi_2$. Since $c = Z a = Z b$ the equation $Z
    (\psi(k)\hadam \xi^A)= c$ has at least the two positive solutions
    $\xi_1$ and $\xi_2$.
  
  \ref{it:t1b}$\Rightarrow$\ref{it:t1c}: 
    Let $k\in\cK_\gamma^+$ and $\xi_1\neq\xi_2\in\RR_{>0}^{n-p}$ as in
    \ref{it:t1b}. For $a=\psi(k)\hadam \xi_1^A$
    and $b=\psi(k)\hadam \xi_2^A=a\hadam\left(\frac{\xi_2}{\xi_1}\right)^A$
    one has $Z a = Z b =c$ by assumption. Hence $Z (a\hadam \xi^A -a)=
    0$ has the positive solution $a=\psi(k)\hadam\xi_1^A$ and
    $\xi=\frac{\xi_2}{\xi_1}$. Further $\xi_1\neq\xi_2$ implies
    $\xi\neq 1$.

  \ref{it:t1c}$\Rightarrow$\ref{it:t1a}: 
    Let $a\in\RR_{>0}^n$ and $\xi\neq 1 \in \RR_{>0}^{n-p}$ be as in
    \ref{it:t1c}. Choose $\lambda\in\Lambda(E)$ and
    let $k=\phi(a^{-1}) \hadam E \lambda$. By Theorem~\ref{thm:k_exists}
    $(k,a)\in V^+$. Let $b=a\hadam \xi^A$. By Lemma~\ref{lem:fibre} also
    $(k,b)\in V^+$. As $Z a = Z b =c$, we have $a, b \in \cP_c$.
\end{proof}

  Theorem~\ref{thm:multi_c_space} item~\ref{it:t1c} is independent of
  the rate constants. That is, given any $c\in\im_+(Z)$ one can decide
  about multistationarity within the corresponding $\cP_c$, even if
  all or some of the rate constants are unknown. In fact, as the
following two corollaries show, 
arbitrary semi-algebraic constraints in the total concentrations $c$
can be added to the description of the multistationarity locus and a
variant of Theorem~\ref{thm:multi_c_space} still holds.  Already the
case of linear inequalities is interesting (see
Section~\ref{subsec:qechambers}).

\begin{cor}
  \label{cor:semi1}
  Assume $V^+$ admits a monomial parametrization with exponent matrix
  $A\in\QQ^{(n-p)\times n}$. Let $g_1,\dots,g_l \in \RR[c]$,
$\square \in\{>,\ge\}^l$, and $\mathcal F(g(c)\ \square \ 0)$ be any
logical combination of the inequalities $g(c)\ \square \ 0$. Then
there are $k\in\cK_\gamma^+$ and $c\in\im_+(Z)$ such that
\begin{displaymath}
Z (\psi(k)\hadam \xi^A) = c, \ \mathcal F (g(c) \ \square \ 0)
\end{displaymath}
has at least two positive solutions $\xi_1 \neq \xi_2$, if and only if
the system
\begin{equation}
  \label{eq:tot_condi}
Z (a \hadam \xi^A-a) = 0,\ \mathcal F(g(Za) \ \square \ 0)
\end{equation}
has a solution $a\in\RR_{>0}^n$ and $\xi\in\RR_{>0}^{(n-p)}$ with
$\xi\neq \1$.
\end{cor}
\begin{proof}
This is Theorem~\ref{thm:multi_c_space} \ref{it:t1b} and \ref{it:t1c}
together with $c=Za$.
\end{proof}

As a consequence of Theorem~\ref{thm:signs_multi}, multistationarity
is possible if and only if the sign condition \eqref{eq:sign_condi}
holds. Frequently one first asks whether multistationarity is possible
at all (by checking condition~(\ref{eq:sign_condi})) before asking
whether it is possible under some conditions on the total
concentrations. Hence, one often computes the
intersection~(\ref{eq:sign_condi}) before employing
Corollary~\ref{cor:semi1}. In this case one can add the information
contained in the sign patterns satisfying~(\ref{eq:sign_condi}) to the
system (\ref{eq:tot_condi}) of Corollary~\ref{cor:semi1}. To this end,
let $\Delta$ be the set of sign patterns satisfying
condition~(\ref{eq:sign_condi}) and recall that, by
  Theorem~\ref{thm:patterns_exist}~(a),
there are $a$, $b\in\RR^n_{>0}$ with
$a\neq b$ and $k\in\cK_\gamma^+$ with $(k,a)\in V^+$, $(k,b)\in V^+$,
and $c=Z a=Z a$, i.e.\ $a,b \in \cP_C$. Then, $\sign(b-a)\in\Delta$ by
Theorem~\ref{thm:signs_multi}.  Further recall that, 
  under the standing assumption 
that $V^+$ admits a monomial parametrization, we have
$b=a\hadam\xi^A$ by Lemma~\ref{lem:mono-connect}. Let
$\delta\in\Delta$, then  
\begin{equation}\label{eq:sign_patterns_a_b}
\sign(b-a) = \delta \Leftrightarrow \sign(\ln b - \ln a) = \delta 
\Leftrightarrow \sign(\xi^A-\1) = \delta.
\end{equation}
Now we can ask whether multistationarity is possible for a given sign
pattern $\delta$ and some semi-algebraic constraint on the total
concentrations:

\begin{cor}
  \label{cor:semi}
  Assume $V^+$ admits a monomial parametrization with exponent matrix
  $A\in\QQ^{(n-p)\times n}$, 
    and that $E$ does not contain a zero row. Let $\Delta$ be the set
    of sign patterns from (\ref{eq:sign_condi}) with
    $\Delta\neq\emptyset$. Pick $\delta\in\Delta$ and
  let $g_1,\dots,g_l \in \RR[c]$,
  $\square \in\{>,\ge\}^l$, and $\mathcal F(g(c) \ \square \ 0)$ be any
  logical combination of the inequalities $g(c)\ \square \ 0$. Then
  there are $k\in\cK_\gamma^+$, $c\in\im_+(Z)$ and $a, b\in
  \RR_{>0}^{n}$ with $a\neq b$ such that
  \begin{displaymath}
    (k,a)\in V^+, (k,b)\in V^+,\quad \sign(b-a) = \delta \text{ and }
    \mathcal F(g(c) \ \square \ 0) 
  \end{displaymath}
  if and only if 
    there are $a\in\RR_{>0}^n$ and $\xi\in\RR_{>0}^{n-p}$ such that
  \begin{equation*}
    Z((\xi^A -\1)\hadam a)=0, \quad
    \sign(\xi^A -\1) = \delta, \quad \mathcal F(g(Za) \ \square \ 0).
  \end{equation*} 
\end{cor}

\begin{proof}
  This is Corollary~\ref{cor:semi1} with $b-a = (\xi^A-\1)\hadam a$ and
  (\ref{eq:sign_patterns_a_b}).
\end{proof}

The next theorem shows that if there is multistationarity for some
value of the total concentrations $c$, then there is also
multistationarity for any rescaled $\alpha c$ ($\alpha>0$), albeit $k$
needs to be adjusted in a nonlinear way.
\begin{thm}
  \label{thm:scaling_c}
  Assume $V^+$ admits a monomial parametrization with exponent matrix
  $A\in\QQ^{(n-p)\times n}$.
    Let $c\in\im_+(Z)$. Then the following are equivalent:
    \begin{enumerate}[label={(\roman*)}]
    \item\label{item:ray1} There exist $a\neq b \in \RR_{>0}^n$ and
      $k\in\RR_{>0}^r$ such that $(k,a)\in V^+$, $(k,b) \in V^+$  and
      $a,b\in\cP_c$.
    \item\label{item:ray2} For every $\alpha>0$ there exists $a(\alpha)\neq
      b(\alpha)\in\RR_{>0}^n$ and $k(\alpha)\in\RR_{>0}^r$ such that 
      $(k(\alpha),a(\alpha))\in V^+$, $(k(\alpha),b(\alpha)) \in V^+$
      and $a(\alpha), b(\alpha) \in \cP_{\alpha c}$. 
    \end{enumerate}
\end{thm}

\begin{proof} 
    Let $a\neq b\in\RR_{>0}^n$ and $k\in\RR_{>0}^r$ as in
    \ref{item:ray1}. By Theorem~\ref{thm:multi_c_space}, item
    \ref{it:t1a}$\Rightarrow$~\ref{it:l3c} there exists
    $a\in\RR_{>0}^n$ and $\xi\neq 1\in\RR_{>0}^{n-p}$ such that
    $Z(a\hadam\xi^A)=Z a= c$. Pick $\alpha>0$ and
    $\lambda\in\Lambda(E)$ and define
    \begin{displaymath}
      a(\alpha) = \alpha a, b(\alpha) = % \alpha a \hadam \xi^A =
      a(\alpha) \hadam \xi^A
      \text{ and } k(\alpha) = \Phi(a(\alpha)^{-1})\hadam E \lambda.
    \end{displaymath}
    Then $(k(\alpha),a(\alpha))\in V^+$ by construction of
    $k(\alpha)$. By Lemma~\ref{lem:fibre} this implies
    $(k(\alpha),b(\alpha))\in V^+$ as well. By construction $Z
    a(\alpha) = Z b(\alpha) = \alpha c$, that is
    $a(\alpha),b(\alpha)\in\cP_{\alpha c}$.

    Vice versa, assume $a(\alpha)\neq b(\alpha)\in\RR_{>0}^n$ and
    $k\in\RR_{>0}^r$ are as in \ref{item:ray2}. Then $\alpha=1$ yields
    the desired result.
\end{proof}

  Theorem~\ref{thm:scaling_c} states that a network $\cN$ for which
  $V^+$ admits a monomial parametrization (and $E$ does not contain a
  zero row) admits multistationarity for some value 
  $c\in\im_+(Z)$, if and only if it admits multistationarity for all
  $\alpha c$ with $\alpha>0$.

The next result can be used to preclude multistationarity on entire
rays in the space of total concentrations.
\begin{cor}
\label{thm:scale}
Assume $V^+$ admits a monomial parametrization with exponent matrix
$A\in\QQ^{(n-p)\times n}$ and let $c\in \im_+(Z)$. If the system
  \begin{equation}
    \label{eq:a_c}
    Z (a \hadam \xi^A) = c
  \end{equation}
  does not have a solution $a\in\RR_{>0}^n$, $\xi\neq 1\in\RR_{>0}^{(n-p)}$,
  then there do not exist $k\in\cK_\gamma^+$ and $\alpha \in \RR_{>0}$
  such that the system
  \begin{equation}
    \label{eq:psi_xi}
    Z (\psi(k)\hadam \xi^A)= \alpha c
  \end{equation}
  has at least two solutions $\xi_1 \ne \xi_2\in\RR_{>0}^p$.
\end{cor}

\begin{proof}
    This is Theorem~\ref{thm:scaling_c} together with
    \ref{thm:multi_c_space}.
\end{proof}

  \begin{remark}
    The scaling invariance in the previous results can be reformulated in
    terms of cones.  For this let $s= \dim ( \im( S ))$ % \cL_{\text{stoi}}$ 
    and denote by $\mathbb S^{n-s-1} \subset \RR^{n-s}$ the unit
    sphere. Define the set of all total concentrations $c \in \im_+(Z)$
    for which the network admits multistationarity (for some value of
    the rate constants $k$):
    \begin{displaymath}
      \cC = \{c\in \im_+(Z)\,|\, \exists k\in \cK_\gamma^+ \text{ and }
      a\ne b\in \RR^n_{>0} \text{ s.t.} (k,a),(k,b) \in V^+, \text{ and } a,b\in\cP_c\}.
    \end{displaymath} 
    By the Tarski--Seidenberg Theorem
    \cite[Theorem~2.3]{coste2002introduction}, $\cC$ is a semi-algebraic
    set.  We have shown that (except the missing origin) it is a cone:
    % if $V^+$ admits a monomial parametrization, then $\cC$ is a cone with
    % the origin removed,
    if $c\in\cC$, then by Theorem~\ref{thm:scaling_c} $\alpha
    c\in\cC $, $\forall \alpha > 0$, i.e.
    \[
      \cC=\left(\cC \cap \mathbb S^{n-s-1}\right)\times\RR_{>0}.
    \]
  \end{remark}

  \begin{remark}
  As a consequence of Theorem~\ref{thm:scaling_c}, for systems that
  admit a monomial parametrization, if there exist rate constants such
  that multistationarity occurs for some value
  $c$ of the total concentrations, then there exist rate constants,
  such that it occurs for arbitrarily small values $\alpha c$, $\alpha
  \ll
  1$. That is, multistationarity persists for arbitrarily small total
  concentrations, as long as the ratios
  $\frac{c_i}{c_j}$ remain constant.
  \end{remark}

\section{Multistationarity conditions on the total
concentrations for sequential and distributive phosphorylation}
\label{subsec:qechambers}

In this section we apply the results of Section~\ref{subsec:nonbin} to
networks describing the sequential and distributive phosphorylation of
a protein. Our results complement recent results
of~\cite{bihan2018lower}.  Their Theorem~4.1 states that, for any
$n\ge 2$, if the total concentration of substrates is greater than the
sum of the concentrations of phosphatase and intermediate products
with phosphatase, then there is a choice of rate constants for which
multistationarity is attained. Our results are also on the
total concentrations and motivate the inequalities using the chamber
decomposition as a natural, intrinsic subdivision of the cone of
values of the total concentrations.

\subsection{Sequential distributive phosphorylation of a protein}
\label{sec:sequ-distr-phosph}

Phosphorylation processes are frequently encountered in the modeling
of biochemical processes; see, for example, \cite{ptm-028} and the
references therein. The following network models the phosphorylation
of a protein $A$ at $n$ binding sites in a sequential and distributive
way:
\begin{equation*}
\label{eq:Npn}
\begin{tikzcd}[row sep=small, column sep=small, every arrow/.append
style={shift left=.75}]
A+K \arrow{rr}{k_{1}} && AK \arrow{rr}{k_{3}} \arrow{ll}{k_{2}} &&
A_{p}+K \\
A_p+P \arrow{rr}{l_{1}} && A_pP \arrow{rr}{l_{3}}
\arrow{ll}{l_{2}} &&
A+P \\
\vdots\\
A_{p^{(n-1)}}+K \arrow{rr}{k_{3n-2}} && A_{p^{(n-1)}}K
\arrow{rr}{k_{3n}} \arrow{ll}{k_{3n-1}} &&
A_{p^{(n)}}+K \\
A_{p^{(n)}}+P \arrow{rr}{l_{3n-2}} && A_{p^{(n)}}P \arrow{rr}{l_{3n}}
\arrow{ll}{l_{3n-1}} && A_{p^{(n-1)}}+P
\end{tikzcd}\tag{$\cN_n$}
\end{equation*}
  The first two connected components of network~\ref{eq:Npn} form
  network~(\ref{eq:fexample}) from Example~\ref{ex:1-site} (after 
  the change of variables $X_1=K$, $X_2=A$, $X_3=A K$, $X_4=A_p$,
  $X_5=P$, $X_6=A_p P$). In this sense \ref{eq:Npn} extends
  \ref{eq:fexample} to $n$ phosphorylation steps.
Due to their biochemical importance such networks have been
extensively studied in mathematical biology.  For example, it is known
that \ref{eq:Npn} is multistationary if and only if $n\ge2$
\cite{holstein2013multistationarity}.  For $n=2$ there are known
sufficient conditions on the rate constants for the presence or
absence of multistationarity and it is known that the number of
positive steady states is $1$, $2$,
or~$3$~\cite{conradi2014catalytic}.  For $n>2$ there are bounds on the
maximum number of positive steady states that can be
attained~\cite{multi-001,BMBN-site}.

The aim of this section is to describe the multistationarity locus in
the space of total concentrations.  The strongest results are
available for the $n=2$ case which we consider first:
  \begin{equation*}
  \label{eq:netN2}
    \begin{tikzcd}[row sep=small, column sep=small, every arrow/.append
      style={shift left=.75}]
      A+K \arrow{rr}{k_{1}} && A K \arrow{rr}{k_{3}} \arrow{ll}{k_{2}} &&
      A_p+K \arrow{rr}{k_4} && A_p K \arrow{rr}{k_6}
      \arrow{ll}{k_5} && A_{pp}+K \\
      A_{pp}+P \arrow{rr}{k_7} && A_{pp} P \arrow{rr}{k_9} \arrow{ll}{k_8} &&
      A_p+P \arrow{rr}{k_{10}} && A_p P \arrow{rr}{k_{12}}
      \arrow{ll}{k_{11}} && A+P
    \end{tikzcd}\tag{$\cN_2$}
  \end{equation*}
  
If all reactions of (\ref{eq:netN2}) are of mass-action form, we
obtain the following set of ODE.
  Here $x_{1}$ denotes the concentration of $A$, $x_{2}$ of $K$,
  $x_{3}$ of $A K$, $x_{4}$ of $A_{p}$, $x_{5}$ of $A_{p} K$, $x_{6}$ of
  $A_{pp}$, $x_{7}$ of $P$, $x_{8}$ of $A_{pp} P$ and $x_{9}$ of
  $A_{p}P$.
\begin{align*}
  \dot{x}_1 &	=f_1(x_1,\ldots,x_9) = -k_1x_1x_2+k_2x_3+k_{12}x_9		      \\
  \dot{x}_2 &	=f_2(x_1,\ldots,x_9) = -k_1x_1x_2+(k_2+k_3)x_3-k_4x_2x_4+(k_5+k_6)x_5  \\
  \dot{x}_3 &	=f_3(x_1,\ldots,x_9) = k_1x_1x_2-(k_2+k_3)x_3			      \\
  \dot{x}_4 & =f_4(x_1,\ldots,x_9) =
              k_3x_3-k_4x_2x_4+k_5x_5+k_9x_8-k_{10}x_4x_7+k_{11}x_9			      \\
  \dot{x}_5 &	=f_5(x_1,\ldots,x_9) = k_4x_2x_4-(k_5+k_6)x_5			      \\
  \dot{x}_6 &	=f_6(x_1,\ldots,x_9) = k_6x_5-k_7x_6x_7+k_8x_8			      \\
  \dot{x}_7 & =f_7(x_1,\ldots,x_9) =
              -k_7x_6x_7+(k_8+k_9)x_8-k_{10}x_4x_7+(k_{11}+k_{12})x_9			      \\
  \dot{x}_8 &	=f_8(x_1,\ldots,x_9) = k_7x_6x_7-(k_8+k_9)x_8			      \\
  \dot{x}_9 & =f_9(x_1,\ldots,x_9) = k_{10}x_4x_7-(k_{11}+k_{12})x_9.  
\end{align*}
There are three independent linear relations among the polynomials
$f_1,\ldots,f_9$ and thus three linearly independent conserved
quantities under the dynamics of the network:
\begin{align}
  \label{eq:consdph}
  \begin{split}
  x_2+x_3+x_5=c_1,\\
  x_7+x_8+x_9=c_2,\\
  x_1+x_3+x_4+x_5+x_6+x_8+x_9=c_3.
  \end{split}
\end{align}
  In (\ref{eq:consdph}) above, $c_1$ represents the total amount of
  kinase $K$, $c_2$ the total amount of phosphatase $P$ and $c_3$ the total
  amount of protein~$A$, the substrate.
Relations \eqref{eq:consdph} are the rays of the \emph{cone of
conservation relations}. According to \eqref{eq:consdph}, we can
choose the conservation matrix as
\begin{equation}
\label{eq:Z}
Z=\left[
  \begin{array}{ccccccccc}
    0 & 1 & 1 & 0 & 1 & 0 & 0 & 0 & 0 \\
    0 & 0 & 0 & 0 & 0 & 0 & 1 & 1 & 1 \\
    1 & 0 & 1 & 1 & 1 & 1 & 0 & 1 & 1				    
  \end{array}
\right].
\end{equation}
In \cite{millan2012chemical} it has been shown that the positive
steady state variety $V^+$ of (\ref{eq:netN2}) admits a monomial
parametrization of the form
\[
  x = \psi(k) \hadam \xi^A \text{ with } k\in\RR_{>0}^{12} \text{ and
  } \xi\in\RR_{>0}^3 \text{ free,}
\]
where
\[
\psi(k) = \left(
\frac{(k_2+k_3)k_4k_6(k_{11}+k_{12})k_{12}}{k_1k_3(k_5+k_6)k_9k_{10}},
\frac{(k_5+k_6)k_9k_{10}}{k_4k_6(k_{11}+k_{12})},
\frac{k_{12}}{k_3},
\frac{k_{11}+k_{12}}{k_{10}},
\frac{k_9}{k_6},
\frac{k_8+k_9}{k_7},
1, 1, 1 \right)^T
\]
and
\begin{equation}\label{eq:A_dd}
A=\left[
  \begin{array}{rrrrrrrrr}
     2 &-1 & 1 & 1 & 0 & 0 & 0 & 0 & 1\\
    -1 & 1 & 0 & 0 & 1 & 1 & 0 & 1 & 0\\
    -1 & 1 & 0 &-1 & 0 &-1 & 1 & 0 & 0
  \end{array}
\right].
\end{equation}

\subsection{A numerical study of multistationarity in the space of
total concentrations}
\label{sec:numerics-dd}

\begin{figure}
\begin{center}
\begin{subfigure}[a]{0.45\textwidth}
\begin{center}
  \includegraphics[height=5.5cm]{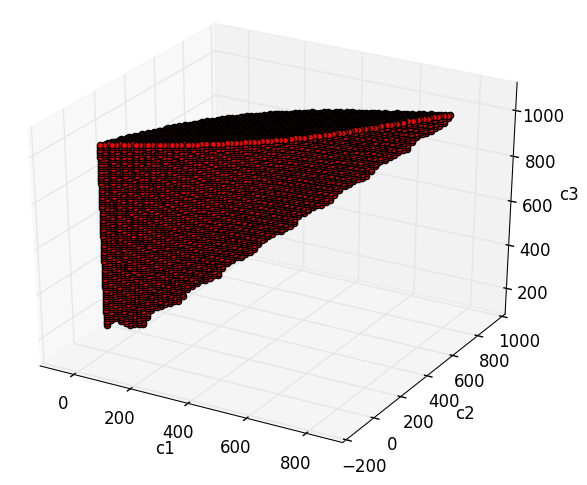}
\end{center}
\caption{Numerical computation with Paramotopy. We consider a grid of
$10^6$ points in the space of total concentrations and represent every
point which leads to multistationarity. The boundary of the
corresponding multistationarity region is represented in red and the
interior in black.  This cone shaped region is semi-algebraic and its
boundary is part of the discriminant in
Fig.~\ref{fig:multistationarityB}.}
\label{fig:multistationarityA}
\end{subfigure} \quad
\begin{subfigure}[a]{0.45\textwidth}
\begin{center}
  \includegraphics[height=6cm]{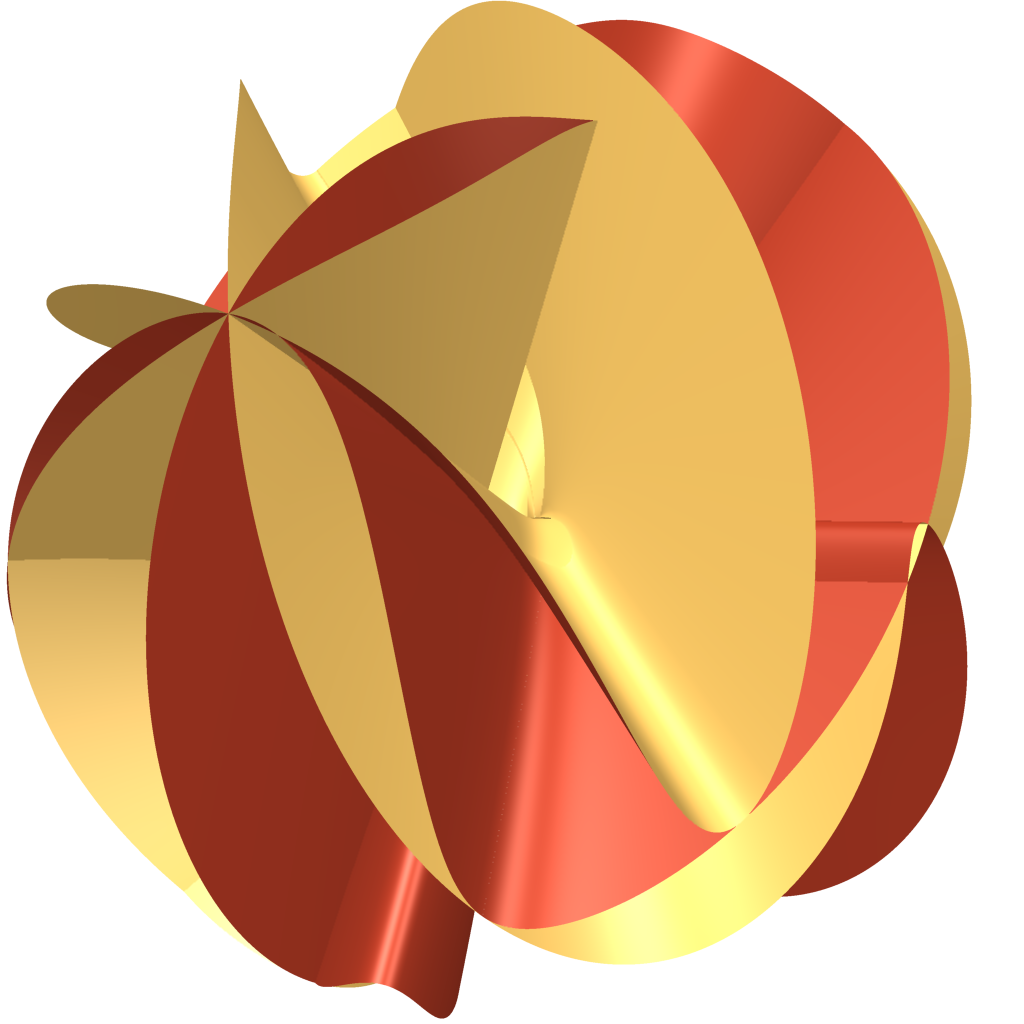}
\end{center}
\caption{The discriminant has seven $\QQ$-ir\-re\-ducible components which
can be found with \Maple.  Three of them are coordinate hyperplanes
and two others are sums of squares.  We show only those two components
which intersect the interior of the positive orthant.  The boundary of
the numerical approximation of the multistationarity region from
Fig.~\ref{fig:multistationarityA} is a subset of this discriminant
surface.}
\label{fig:multistationarityB}
\end{subfigure}
\caption{
\label{fig:multistationarity} 
Representation of regions of multistationarity in the space of
total concentrations for \ref{eq:netN2}. For both figures all rate constants
have been fixed to the values given in
\cite[Fig.~3]{conradi2014catalytic}.
}
\end{center}
\end{figure}

We did a numerical study of multistationarity in the space of total
concentrations which is depicted in
Fig.~\ref{fig:multistationarityA}. For this computation the rate
constants have been numerically fixed to the values in
\cite[Fig.~3]{conradi2014catalytic}.  The computation was done using
Paramotopy~\cite{paramotopy} which builds on Bertini~\cite{bertini}
and allows one to efficiently analyze the solutions 
  of a polynomial systems with unknown coefficients (a parametric
  polynomial system).
 We computed the isolated solutions for each point
in the grid $[0,1000]^3\cap (10\ZZ)^3$ and plotted those which yield
multistationarity.  An alternative approach is through the
\emph{discriminant} which in this case can be found with
\Maple~\cite{maple}.  A discriminant of a parametric semi-algebraic
system is a polynomial which vanishes in those points of the
parameter space where the solution behavior can change.
For an extensive discussion of discriminants with a special emphasis
on computation we refer to~\cite{LAZARD2007636}. Discriminants for
multistationarity have also appeared in \cite[Section~4]{Gross2016}.
Two relevant 
irreducible components of the discriminant of the parametric system
are visualized in Fig.~\ref{fig:multistationarityB}.  The algebraic
boundary of the region from Fig.~\ref{fig:multistationarityA} is a
subvariety of the discriminant from Fig.~\ref{fig:multistationarityB}.
Specifically, the cone shaped region in
Fig.~\ref{fig:multistationarityA} is also visible in the top center of
Fig.~\ref{fig:multistationarityB}. Both figures indicate that, for the
values of the rate constants chosen
in~\cite[Fig.~3]{conradi2014catalytic}, multistationarity does not
occur for all values of the total concentrations. In the next section
we employ the results of Section~\ref{sec:totConc} 
% together with the
% chamber decomposition of Section~\ref{sec:con-rel}
to elucidate
conditions on the total concentrations for the presence or absence of
multistationarity.

\subsection{The chamber decomposition of \texorpdfstring {$\im_+(Z)$}{im+(Z)} for \ref{eq:netN2}
  and \ref{eq:Npn}}
\label{sec:chamber-dd}

  Using the sets $\im_+(Z)$ and $\cP_c$ from (\ref{eq:def_imp_Z}) and
  (\ref{eq:parametricPoly}), we now introduce the chamber
  decomposition of $\im_+(Z)$ induced by the columns of the
  conservation matrix~$Z$.  We assume that $Z$ is of full row rank
  $n-s$ and call a subset of $n-s$ linear independent columns a basis
  of $\im_+(Z)$. 
  Each basis $B$ defines a \emph{basic cone} $\cone (B)$ consisting of
  nonnegative linear combinations of the columns in~$B$.
  \begin{defn}\label{d:chamberComplex}
    The \emph{chamber complex} of a matrix $Z$ is the common refinement of
    the basic cones of all its bases.  More precisely, $c_1$ and $c_2$ are
    in the same chamber of the chamber complex if and only if
    \[
      c_1 \in \cone (B) \iff c_2\in \cone(B) \qquad \text{for all bases $B$
        of $Z$}.
    \]
  \end{defn}
  \begin{remark}
    The chamber complex is important in linear programming as it
    classifies the different combinatorial types that the
    polyhedron $\cP_c$
    can take for any $c$: within one chamber, all
    polyhedra $\cP_c$ are combinatorially equal, that is, their face
    lattices are the same.  See \cite[Section~2.1]{de2009graphs} for an
    interpretation of the vertices of $\cP_c$ in this context.  Chamber
    complexes can be computed with polyhedral geometry software such as
    \textsc{Polymake}~\cite{GawrilowJoswigPolymake} or
    \textsc{TOPCOM}~\cite{Rambau:TOPCOM-ICMS:2002}.  Chambers are also
    related to siphons of chemical reactions~\cite{alg-030}.  A
    chamber complex of a slightly different type appears in
    \cite{craciun2009algebraic} where every basic cone encodes a
    possible reaction network among a given finite set of
    experimentally indistinguishable networks.  We believe that the
    chamber complex is an interesting structure to study for
    different chemical reaction networks.
  \end{remark}

The polyhedron $\cP_c$ is defined by the matrix $Z$ from
(\ref{eq:Z}). The cone generated by the columns of $Z$ is $\RR^3_{\ge
0}$. There are eight basic cones generated by the following sets of
columns of $Z$:
\[
\{1,2,7\}, \{1,2,8\}, \{1,3,7\}, \{1,3,8\}, \{2,3,7\}, \{2,3,8\},
\{2,7,8\}, \{3,7,8\}.
\]
Any of the basic cones is the intersection of three linear half-spaces
of~$\RR^3$ and 
  each of these half-spaces is
spanned by exactly two of the
three columns (see~\cite[Section~1.1]{ziegler2012lectures} for more
details on polyhedra). For example, the cone generated by the columns
of $\{1,2,7\}$ of $Z$ is $\RR^3_{\ge0}$ and equals the intersection of
the half-spaces $c_1\ge0$, $c_2\ge0$, and $c_3\ge0$. There are six
distinct planes occurring among the defining hyperplanes of all cones:
$c_1=0$, $c_2=0$, $c_3=0$, $c_1=c_3$, $c_2=c_3$, and $c_1+c_2=c_3$.
These planes divide $\RR^3_{\ge 0}$ into five full-dimensional
cones. The interiors of these cones are the full-dimensional chambers
of $\RR^3_{\ge 0}$.  See Fig.~\ref{fig:chambers} for a two-dimensional
representation of this chamber decomposition.  There are 
also smaller dimensional chambers: the interiors of the faces of the
full-dimensional chambers. %As it turns out, the whole above analysis
% extends beyond the network $\cNtwo$ and 

  For $n\geq 2$ the chamber complex does not change:

\begin{thm} 
The cone of conservation relations of \ref{eq:Npn} is $\RR^3_{\ge0}$
and it has five full-dimensional chambers: \label{thm:chambers}
  \begin{equation*}
    \begin{array}{lll}
      \Omega(1):
      \begin{cases}
        c_3>0\\
        c_2>c_3\\
        c_1>c_3,
      \end{cases}
   &
     \Omega(2):
     \begin{cases}
       c_1>0\\
       c_2>c_3\\
       c_1<c_3,
     \end{cases}
   &
     \Omega(3):
     \begin{cases}
       c_2>0\\
       c_2<c_3\\
       c_1>c_3,
     \end{cases}\\
      \Omega(4):
      \begin{cases}
        c_1<c_3\\
        c_2<c_3\\
        c_1+c_2>c_3,
      \end{cases}
   &
     \Omega(5):
     \begin{cases}
       c_1>0\\
       c_2>0\\
       c_1+c_2<c_3.
     \end{cases}
    \end{array}
  \end{equation*}
\end{thm}

\begin{proof}
As described in \cite[Section~3]{holstein2013multistationarity}, the
conservation matrix of~\ref{eq:Npn}, for the ordering of the
concentrations defined in
\cite[Table~1]{holstein2013multistationarity}, has the form
$Z^{(n)}=(Z_0|Z_1|\ldots|Z_1)\in\RR^{3\times(3n+3)}$, where $Z_1$ is
repeated $n$ times and
\begin{equation*}
\begin{array}{cc}
  Z_0=\left[
  \begin{array}{ccc}
    1&0&0\\
    0&0&1\\
    0&1&0\\
  \end{array}\right], \quad
  Z_1=\left[
  \begin{array}{ccc}
    1&0&0\\
    0&0&1\\
    1&1&1\\
  \end{array}\right].
\end{array}
\end{equation*}
As $Z^{(n)}$ has the same set of columns for every $n \ge 1$, it
follows that all chamber decomposition of all~\ref{eq:Npn} are equal.
\end{proof}

\begin{remark} Although the ordering of variables defined in
\cite[Table~1]{holstein2013multistationarity} is different from the
one we use with \ref{eq:netN2}, a reordering of the variables
corresponds to a reordering of the columns of $Z^{(n)}$ and thus, it
leaves the chamber decomposition invariant.
\end{remark}

\begin{remark}
Although \ref{eq:Npn} has the same chamber complex for each $n$, the
constants $c$ express nonnegative linear combinations of the
concentrations specific to each network.
\end{remark}

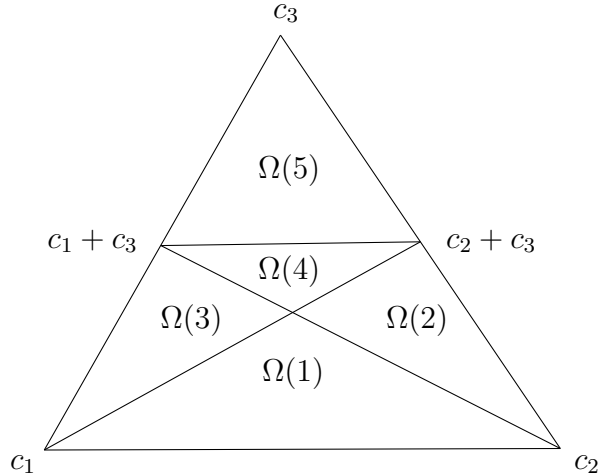
\begin{figure}
\begin{tikzpicture}
\clip(0.37,-5.68) rectangle (10.31,2.6); \draw (5.24,1.44)--
(2.1,-4.08); \draw (5.24,1.44)-- (8.96,-4.06); \draw (2.1,-4.08)--
(8.96,-4.06); \draw (3.65,-1.36)-- (7.1,-1.31); \draw (2.1,-4.08)--
(7.1,-1.31); \draw (3.65,-1.36)-- (8.96,-4.06); \draw (1.5,-4)
node[anchor=north west] {$c_1$}; \draw (9,-4) node[anchor=north west]
{$c_2$}; \draw (5,2) node[anchor=north west] {$c_3$}; \draw (2,-1)
node[anchor=north west] {$c_1+c_3$}; \draw (7.3,-1) node[anchor=north
west] {$c_2+c_3$}; \draw (4.85,-2.7) node[anchor=north west] {
$\Omega(1)$}; \draw (6.5,-2) node[anchor=north west] {$\Omega(2)$};
\draw (3.5,-2) node[anchor=north west] {$\Omega(3)$}; \draw
(4.8,-1.35) node[anchor=north west] {$\Omega(4)$}; \draw (4.8,0)
node[anchor=north west] {$\Omega(5)$};
\end{tikzpicture}
\caption{The intersection of the full-dimensional chambers associated
to \ref{eq:Npn} with the plane $c_1+c_2+c_3=1$. Labeled
vertices correspond to different columns in $Z^{(n)}$.}
\label{fig:chambers}
\end{figure}

\subsection{Multistationarity conditions in the space of total
  concentrations}
\label{sec:new-multi}

Now we 
  turn to $n=2$ and \ref{eq:netN2} and
employ Corollary~\ref{cor:semi} to decide whether
multistationarity is possible for total concentrations in the chambers
$\Omega(i)$.  The linear inequality conditions $c\in\Omega(i)$ become
the conditions $\mathcal F( \bullet )$ in Corollary~\ref{cor:semi}.
We also integrate the information in the sign patterns in the
intersection~\eqref{eq:sign_condi}.  These have been computed in
\cite{conradi2008multistationarity} and are encoded as rows of the
following matrix $\Delta$ (or their negatives):
\begin{equation}
  \label{eq:signs}
  \Delta=\left[
    \begin{array}{rrrrrrrrr}
      -1&-1&-1& 1& 1& 1&-1& 1&-1\\
      1& 0& 1&-1&-1&-1& 1&-1& 1\\
      1&-1& 1& 1&-1&-1& 1&-1& 1\\
      1&-1& 1& 1&-1&-1& 0&-1& 1\\
      1&-1& 1& 1&-1&-1&-1&-1& 1\\
      1&-1& 1& 0&-1&-1& 1&-1& 1\\
      1&-1& 1&-1&-1&-1& 1&-1& 1\\
    \end{array}
  \right].
\end{equation}
The rows $\delta_i$ of $\Delta$ define the conditions
$\sign(\xi^A-\1)=\delta_i$, $i=1,\dots, 7$ of
Corollary~\ref{cor:semi}. Using the matrix $A$ from \eqref{eq:A_dd},
this condition reads

\begin{equation}
\label{eq:signxi} 
\text{sign} \left(
  \frac{\xi^2_1}{\xi_2\xi_3}-1, \frac{\xi_2\xi_3}{\xi_1} -1, \xi_1 -1,
  \frac{\xi_1}{\xi_3} -1, \xi_2 -1, \frac{\xi_2}{\xi_3} -1, \xi_3 -1,
  \xi_2 -1, \xi_1-1 \right) = \delta_i.
\end{equation}

To check multistationarity for $c$ in all of the chambers $\Omega(i)$
we use \Mathematica~\cite{Mathematica}. For each chamber and each row
$\delta_i$ we set up the conditions of Corollary~\ref{cor:semi} and
use the command \texttt{Reduce} to decide the existence of
solutions. In the following example we show how to set up the
\Mathematica code to check multistationarity.
% in $\Omega(1)$ for
% $\delta_1$, the first row of $\Delta$.

\begin{example}\label{ex:HowToChamber}
    We check multistationarity via Corollary~\ref{cor:semi} for
    $\Omega(1)$ and $\delta=($ $1$, $1$, $1$, $-1$, $-1$, $-1$, $1$,
    $-1$, $1$ $)$ (the negative of the first row of $\Delta$ from
    (\ref{eq:signs})). First we formulate the three conditions
    $Z ((\xi^A-1)\hadam a)=0$, $\sign(\xi^A-1)=\delta$ and
    $\mathcal{F}(g(Za)) \ \square \ 0$ 
    of the corollary:
    \begin{itemize}
    \item $\sign(\xi^A-1)=\delta$: after adding the constraints
      $\xi_1>0$, $\xi_2>0$, $\xi_3>0$ and removing redundant
      inequalities, \eqref{eq:signxi} reduces to
      \begin{equation}
        \label{eq:condi1}
        1<\xi_1<\xi_2\xi_3<\xi^2_1 \text{ and } 0<\xi_2<1 .
      \end{equation}
    \item $\mathcal{F}(g(Za) \ \square \ 0)$: we want to encode
      $c\in\Omega(1)$. By Theorem~\ref{thm:chambers} this is equivalent
      to $c_3>0$, $c_1>c_3$ and $c_2>c_3$. By (\ref{eq:consdph}),
      $c_1=a_2+a_3+a_5$, $c_2= a_7+a_8+a_9$ and
      $c_3=a_1+a_3+a_4+a_5+a_6+a_8+a_9$. The condition $a_i>0$ then
      implies $c_3>0$ and we only need to account for the remaining
      inequalities:
      \begin{equation}
        \begin{split}
          \label{eq:condi2}
          a_1 - a_2 +a_4 +a_6+ a_8+a_9 <0 &\text{ (for $c_1>c_3$) and } \\
          a_1+a_3+a_4+a_5+a_6-a_7 < 0 &\text{ (for $c_2>c_3$).}
        \end{split}
      \end{equation}
    \item In the condition $Z((\xi^A-\1)\hadam a)=0$ we use the matrix  
      \begin{equation}
        \label{ex:Z'}
        Z'=\left[
          \begin{array}{rrrrrrrrr}
            0 &  1 & 1 & 0 & 1 & 0 &  0 & 0 & 0\\
            1 &  0 & 1 & 1 & 1 & 1 & -1 & 0 & 0\\
            1 & -1 & 0 & 1 & 0 & 1 &  0 & 1 & 1
          \end{array}\right],
      \end{equation} 
      obtained from~\eqref{eq:Z} by elementary row operations
      (compared to $Z$ from (\ref{eq:Z}), we found that computation
      times are significantly shorter when $Z'$ is used,
      cf.~Remark~\ref{rem:QE_I}). To obtain polynomial conditions 
      ($\xi^A$ is rational) we write this condition as 
      $Z' (\xi^A \hadam a)= Z' a$ and clear denominators:
      \begin{equation}
        \begin{split}
          \label{eq:condi3}
          \xi_2 \xi_3 a_2 + \xi_1^2 a_3 + \xi_1 \xi_2 a_5 &=
          \xi_1 (a_2+a_3+a_5) \\ 
          \xi_1^2 a_1 + \xi_1 \xi_2 \xi_3 a_3 + \xi_1 \xi_2 a_4 +
          \xi_2^2 \xi_3 a_5  + \xi_2^2 a_6 - \xi_2 \xi_3^2 a_7 &=
          \xi_2 \xi_3 (a_1+a_3+a_4+a_5+a_6-a_7) \\ 
          \xi_1^3 a_1 - \xi_2^2 \xi_3^2 a_2 + \xi_1^2 \xi_2 a_4 +
          \xi_1 \xi_2^2 a_6 + \xi_1 \xi_2^2 \xi_3 a_8 +
          \xi_1^2 \xi_2 \xi_3 a_9 &=
          \xi_1 \xi_2 \xi_3 (a_1-a_2+a_4+a_6+a_8+a_9) 
        \end{split}
      \end{equation}
    \end{itemize}
    The following \Mathematica code can be used to decide the
    existence of $\xi_i$ and $a_i$ satisfying the conditions
    (\ref{eq:condi1}), (\ref{eq:condi2}) and (\ref{eq:condi3})
    together with the condition $a_i>0$ (x1, x2, x3 are
    shorthand for the variables $\xi_1,\xi_2, \xi_3$):
\begin{Verbatim}
Reduce[Exists[{a1,a2,a3,a4,a5,a6,a7,a8,a9},
a1>0 && a2>0 && a3>0 && a4>0 && a5>0 && a6>0 && a7>0 && a8>0 && a9>0 &&
x2*x3*a2 + x1^2*a3 + x1*x2*a5 == x1*(a2+a3+a5) && 
x1^2*a1 + x1*x2*x3*a3 + x1*x2*a4 + x2^2*x3*a5 
   + x2^2*a6 - x2*x3^2*a7 == x2*x3*(a1+a3+a4+a5+a6-a7) && 
x1^3*a1 - x2^2*x3^2*a2 + x1^2*x2*a4 + x1*x2^2*a6 
   + x1*x2^2*x3*a8 + x1^2*x2*x3*a9 == x1*x2*x3*(a1-a2+a4+a6+a8+a9) &&
a1+a3+a4+a5+a6-a7<0 && 
a1-a2+a4+a6+a8+a9<0 &&
x2*x3<x1^2 && x1<x2*x3 && 
1<x1 && 1>x2 && x2>0 ]]
\end{Verbatim}
    The computation takes a few hours, but then the result is
  \lq False\rq, that is, there do not exist $a_1,\ldots,a_9$
  satisfying the constraints, no matter what the values of
  $\xi_1,\xi_2,\xi_3$ are. Consequently, in the chamber $\Omega(1)$
  there is no multistationarity coming from the first row of $\Delta$.
  Theorem~\ref{thm:table} below shows that there is no
  multistationarity in $\Omega(1)$ at all.
\end{example}

Theorem~\ref{thm:table} spells out for which chambers and which sign
patterns there is multistationarity.  For a pair $(\Omega(i),\delta_j)$, we
write $+$ if there is multistationarity in $\Omega(i)$ for all values
of $\xi$ compatible with~\eqref{eq:signxi}. We write $++$ if there is
multistationarity in $\Omega(i)$ with extra conditions for $\xi$
stronger than~\eqref{eq:signxi}.  We write $-$ if there is no
multistationarity.  If the computation does not finish in
  reasonable time, we leave the cell empty.

\begin{table}[htpb]
  \begin{tabular}{ | c || c | c | c | c | c | c | c | }
    \hline
    &$\delta_1$&$\delta_2$&$\delta_3$&$\delta_4$&
    $\delta_5$&$\delta_6$&$\delta_7$\\\hline\hline
    $\Omega(1)$&--&--&--&--&--&--&--\\\hline
    $\Omega(2)$&--&--&+&+&++&+&++\\ \hline
    $\Omega(3)$&++&+&++&--&--&+&+\\ \hline
    $\Omega(4)$&++&+& &+&++& & \\ \hline
    $\Omega(5)$&++&+&+&+&+&+&+\\ \hline
  \end{tabular}
  \caption{
    \label{fig:chambers-signs}
    The chamber-signs incidence table of (\ref{eq:netN2}).  In particular,
    multistationarity is not possible in $\Omega(1)$.  
  }
\end{table}

\begin{thm}
\label{thm:table}
Up to the three empty cells, the chambers-signs incidence table of
(\ref{eq:netN2}) is Table~\ref{fig:chambers-signs}. For the $++$ entries the
following additional constraints are derived:
\begin{equation*}
\begin{array}{ll}
  (\Omega(2),\delta_7):&0<\xi_3<\xi_1<1 \ \land \ \xi_2>\frac{\xi_1^2}{\xi_3^2},\\
  (\Omega(2),\delta_5):&\xi_3>1 \ \land \ 0<\xi_1<1 \ \land \ \xi_2>\xi_3^2,\\
  (\Omega(4),\delta_5):&\xi_3>1 \ \land \ 0<\xi_1<1 \ \land \
			 \xi_2>\xi^2_3,\\
  (\Omega(3),\delta_3):&\xi_3^2<\xi_1<\xi_3<1 \ \land \xi_2>1,\\
  (\Omega(3),\delta_1):&\xi_3>1 \ \land \\ & \left(\left(1<\xi_1<\xi_3^{2/3} \
					   \land \ \frac{\xi_1}{\xi_3}<\xi_2<{\frac{\xi_1^{3/2}}{\xi_3}}\right)
					      \lor\left(\xi_3^{2/3}<\xi_1<\xi_3 \ \land \ \frac{\xi_1}{\xi_3}<\xi_2<1
										     \right)\right),\\
  (\Omega(4),\delta_1):&\xi_3>1 \ \land \\ & \left(\left(1<\xi_1<\xi_3^{2/3} \
					   \land \ \frac{\xi_1}{\xi_3}<\xi_2<{\frac{\xi_1^{3/2}}{\xi_3}}\right)
					      \lor\left(\xi_3^{2/3}<\xi_1<\xi_3 \ \land \ \frac{\xi_1}{\xi_3}<\xi_2<1
										     \right)\right),\\
  (\Omega(5),\delta_1):&\xi_3>1 \ \land \\ &  \left(\left(1<\xi_1<\xi_3^{1/2}
					   \ \land \ \frac{\xi_1}{\xi_3}<\xi_2<\frac{\xi_1^2}{\xi_3}\right)\lor
						\left(\xi_3^{1/2}<\xi_1<\xi_3 \ \land \ \frac{\xi_1}{\xi_3}<\xi_2<1 \right)\right).\\
\end{array}
\end{equation*}
\end{thm}
\begin{proof}[Computational Proof]
The quantifier elimination problems were set up similarly to
Example~\ref{ex:HowToChamber} and solved using \Mathematica.
\end{proof}

\begin{remark}\label{r:indirect}
To obtain Table~\ref{fig:chambers-signs}, some of the computations
were made indirectly. For example, we checked that for $\delta_1$
multistationarity doesn't take place in $\Omega(1)$ but we couldn't
check directly that it doesn't take place in $\Omega(2)$, as
the computations did not finish within one to five days.
We therefore checked that it does not take place in
$\Omega(1) \cup \Omega(2) \cup \Omega(1,2)$, where $\Omega(1,2)$
denotes the boundary between $\Omega(1)$ and $\Omega(2)$.  This
computation was feasible. It is an interesting computational
challenge to classify all boundaries between chambers.
\end{remark}

\begin{remark}\label{rem:QE_I}
The quantifier elimination problems arising from the analysis of
multistationarity have additional structure that should be exploited.
In particular, the run times of our computations seem to be sensitive
to the formulation of the input.  We experimented with different
equivalent semi-algebraic systems in \Mathematica.  One knob to turn is
the system $Z((\xi^{A}-\1)\hadam a)=0$ in Corollary~\ref{cor:semi}.
Different bases for the row space of $Z$ lead to different run times.
Consider the pair $(\Omega(4),\delta_4)$ and let $R_1$, $R_2$, and
$R_3$ be the rows of~$Z$ from eq.~(\ref{eq:Z}). Let
$Z_1=\left[ (R_1+R_2+R_3)^T| (R_2+R_3)^T| R_3^T \right]^T$,
$Z_2=\left[ R_1^T| R_2^T| (R_3-R_1-R_2)^T \right]^T$, and
$Z_3=\left[R_1^T|(R_3-R_2)^T|(R_3-R_1)^T\right]$.  Using in
Corollary~\ref{cor:semi} the matrix $Z$ from (\ref{eq:Z}), the
computation takes about seven seconds while with either of $Z_1$,
$Z_2$, and $Z_3$ the computation did not finish within 24 hours.  It
is tempting to think that the computations with the matrix $Z$ are
faster because it is in row echelon form; however this is not the
case: for the pair $(\Omega(1),\delta_1)$ the computation with $Z$ did
not finish in several days while the computation with $Z_3$ finished
within a few hours.
\end{remark}

\begin{remark}
\label{rem:QE_II} Since in Corollary~\ref{cor:semi} we are only
interested in the positive solutions of the system
$Z((\xi^A-\1)\hadam x)=0$, clearing denominators does not add any new
solutions. Let $\varsigma(Z,\xi^A,\delta,x)$ denote the system
obtained from $Z((\xi^A-\1)\hadam x)=0$ and $\delta$, by clearing
denominators. If $Z'$ and ${A}'$ are matrices obtained by performing
elementary row operations on $Z$ and $A$ respectively, then
$\varsigma(Z,\xi^A,\delta,x)$ and $\varsigma(Z',\xi^{A'},\delta,x)$
have the same set of positive solutions (they are \emph{equivalent
systems}), yet they are not linearly equivalent systems.
\end{remark}

\begin{remark}
\label{rem:QE_III}
Throughout we found \Mathematica to have the fastest implementation of
quantifier elimination.  It would be nice to implement heuristics for
pre-simplification, e.g.\ along the lines of~\cite{brown2010black}, in
open source systems such as \texttt{qepcadB}~\cite{brown2003qepcad} or
\texttt{REDLOG}~\cite{dolzmann1997redlog}.  The performance of
quantifier elimination on systems from biology has been explored
in~\cite{BDEEGGHKRSW2017}.
\end{remark}

The first row of Table~\ref{fig:chambers-signs} shows that
multistationarity is only possible if $c \notin \Omega(1)$:

\begin{cor}
\label{theo:no_multi_dd}
For (\ref{eq:netN2}), if $c \in \Omega(1)$, then there is no
$k\in\RR^{12}_{>0}$ such that the equations
\begin{displaymath}
S v(k,x) =0,\,\, Z x = c
\end{displaymath}
have at least two positive solutions. 
\end{cor}

Together with results of Bihan, Dickenstein, and Giaroli we
almost obtain a characterization of multistationarity for 2-site
phosphorylation in the total concentration coordinates.  The regions
of multistationarity are polyhedral and the only unresolved cases are
when $c_2=c_3$, $c_1=c_3$, or both.
\begin{cor}\label{cor:MM-setting}
In the 2-site phosphorylation network, multistationarity is impossible
if $c_3<c_2$ and $c_3<c_1$, and possible if $c_3>c_2$ or $c_3>c_1$.
If $S_{\text{tot}}$ is the total concentration of substrate,
$F_{\text{tot}}$ that of phosphatase, and $E_{\text{tot}}$ that of
kinase, then multistationarity is impossible if
$S_{\text{tot}} < E_{\text{tot}}$ and
$S_{\text{tot}} < F_{\text{tot}}$ and possible if
$ S_{\text{tot}} > E_{\text{tot}}$ or
$S_{\text{tot}} > F_{\text{tot}}$.
\end{cor}
\begin{proof}
Theorem~4.1 in \cite{bihan2018lower} says (in our notation) that if
$c_3>c_2$, then there is a choice of rate constants $k$ for which there
is multistationarity.  There is an inherent symmetry of the system,
exchanging $x_2, x_3, x_5$ with, respectively, $x_7, x_8, x_9$.  Under
this symmetry, the mathematical properties are unchanged, but $c_1$
and $c_2$ change roles.  Therefore there is multistationarity also if
just $c_3>c_1$.  Theorem~\ref{thm:table} shows that if both $c_1<c_3$
and $c_2 < c_3$ then multistationarity is impossible.
\end{proof}

To characterize multistationarity for \eqref{eq:netN2}, the
following boundary cases remain
\begin{itemize}
\item $c_3=c_2$ and $c_1<c_3$,
\item $c_3<c_2$ and $c_1=c_3$,
\item $c_3=c_2$ and $c_1=c_3$.
\end{itemize}
We attempted this classification using computations as in
Example~\ref{ex:HowToChamber}.  For several combinations of signs and
items above, we could rule out multistationarity, but no conclusions
were possible.  One region we could not rule out was sign $\delta_7$
combined with $c_1 = c_2 = c_ 3$.  For this instance we employed the
numerical solver \scip~\cite{GleixnerEtal2018OO}.  It could not find a
solution to the corresponding inequality system and the run on the
computation indicates that multistationarity is impossible in this
region too.  From these computational experiences we conjecture that
multistationarity is impossible in the boundary regions.

\section{On the existence of monomial parametrizations for  
  \texorpdfstring {$V^+$}{V+}}
\label{sec:existence}

Definition~\ref{def:mono} uses the strictly positive steady states.
This differs from the definition of \emph{toric steady states}, which
uses the steady state ideal and thus poses restrictions on all complex
solutions of the steady state equations.
Example~\ref{ex:monomialparam} demonstrates the, maybe unsurprising,
fact that the positive real part can have a monomial parametrization
while the whole steady state variety does not.  We discuss these
phenomena in the context of decompositions of binomial
ideals~\cite{ES96,kahle11mesoprimary}.

It follows from \cite[Corollary~1.2]{ES96} that a binomial Gröbner
basis of the steady state ideal is sufficient for toric steady states
and thus a monomial parametrization of the positive steady states (by
Proposition~\ref{prop:binomcomp}).  A binomial steady state ideal,
however, is not necessary for this.  The steady state ideal may
possess primary components that are irrelevant to the positive real
part.  We first illustrate this fact with an example.

\begin{example}\label{ex:monomialparam}
Let $\cT$ be the following triangular network
\cite[Example~2.3]{millan2012chemical}:
\begin{equation*}
\begin{tikzcd}[row sep=huge, every arrow/.append style={shift left=.75}]
& 2X_1 \arrow{dr}{1} \arrow{dl}{1} \\
2X_2 \arrow{ur}{1} \arrow{rr}{1} && X_1+X_2 \arrow{ul}{1}
\arrow{ll}{1}
\end{tikzcd}\tag{$\cT$}
\end{equation*}
Let $x_i$ denote the concentration of $X_i$. The steady state ideal of
network $\cT$ is
$I_1 = \langle x_1^2-x_2^2 \rangle = \langle x_1-x_2 \rangle \cap
\langle x_1+x_2 \rangle$.  The Zariski closure of the positive steady
state variety $\overline{V^+}_{1}=\VV(x_1-x_2)$ has exactly one
irreducible component defined by one binomial and is thus a toric
variety.  It has a monomial parametrization $x_1=x_2=s$, for
$s\in\RR$.  Restricting this monomial parametrization to the interior
of the positive orthant yields a parametrization for $V_1^+$ (see
Fig.~\ref{fig:binomvarietyA}).  Let
\begin{equation}
\label{eq:example_ideal}
\begin{array}{ll}
  I_2& = I_1 \cap \langle x_1+x_2+1 \rangle = \langle x_1 - x_2
       \rangle \cap \langle x_1 + x_2 \rangle \cap \langle x_1 + x_2 + 1
       \rangle \\
     & = \langle -x_1^3 -x_1^2x_2+x_1x_2^2+x_2^3-x_1^2+x_2^2 \rangle.
\end{array}
\end{equation} 
Clearly, $I_2$ is not binomial; $I_2$ is the intersection of two prime
binomial ideals and a prime trinomial ideal.  Geometrically, the
intersection of ideals corresponds to taking the union of the
corresponding varieties as in Fig.~\ref{fig:binomvarietyB}.  Only the
component $\VV(x_1-x_2)$ of $\VV(I_2)$ intersects the interior of the
positive orthant.  Still, $I_2$ can be the steady state ideal of some
mass-action network.  According to
\cite[Section~4.7.1.1]{erdi1989mathematical}, a mass-action network is
described by a system of ODEs of the form $\dot x =f$, where
$f \in \RR[x]^n$, if and only if every negative term in $f_i$ is
divisible by the variable $x_i$. This condition is fulfilled by the
following system of ODEs:
\begin{equation*}
\dot x_1 = -\dot x_2 = -x_1^3-x_1^2x_2+x_1x_2^2+x_2^3-x_1^2+x_2^2.
\end{equation*}
One network whose steady state ideal is equal to $I_2$ is $\cS$:
\begin{equation*}
\begin{tikzcd}[row sep=huge, every arrow/.append style={shift left=.75}]
3X_1 \arrow{rr}{2/3} \arrow{d}{1/9} && 2X_1+X_2 \arrow{ll}{1} \arrow{d}{2} \\
3X_2 \arrow{u}{1/9} \arrow{rr}{2/3} && X_1+2X_2 \arrow{ll}{1}
\arrow{u}{2}
\end{tikzcd} \qquad
\begin{tikzcd}[row sep=huge, every arrow/.append style={shift left=.75}]
 2X_1 \arrow{rr}{1/2}  &&2X_2 \arrow{ll}{1/2}
\end{tikzcd} \tag{$\cS$}
\end{equation*}
Summarizing, the steady state variety $\VV(I_2)$ has three irreducible
components, but only $\VV( x_1 - x_2)$ intersects the interior of the
positive orthant.  Since $V^+_{1}=V^+_{2}$, the positive steady
state varieties of $\cT$ and $\cS$ share the parametrization
$x_1=x_2=s$, for $s\in\RR_{>0}$.

\begin{figure}
\begin{subfigure}[a]{0.4\textwidth}
\begin{tikzpicture}
\draw[->] (-2,0) -- (2,0) node[right] {$x_1$};
\draw[->] (0,-2) -- (0,2) node[above] {$x_2$};
\draw[scale=0.5,domain=-3:3,smooth,variable=\x,black]
plot ({\x},{\x});
\draw[scale=0.5,domain=-3:3,smooth,variable=\x,black]
plot ({\x},{-\x});
\node at (1.5,1.8) {$\VV(x_1-x_2)$};
\node at (-1.5,1.8) {$\VV(x_1+x_2)$};
\node at (1.5,-2) {\phantom{$\VV(x_1+x_2+1)$}};
\end{tikzpicture}
\caption{The variety $\VV(x_1^2-x_2^2)=\VV(x_1-x_2)\cup \VV(x_1+x_2)$
of $\cT$ from Example~\ref{ex:monomialparam}. $\cS$ has toric
steady states as its steady state ideal is binomial and
$\overline{V^+_1}$ is nonempty and irreducible (see
\cite[Definition~2.2]{millan2012chemical}). $V^+_1$ is
parametrized by $s \mapsto (s,s)$, for $s\in\RR_{>0}$.}
\label{fig:binomvarietyA}
\end{subfigure} \quad
\begin{subfigure}[a]{0.4\textwidth}
\begin{tikzpicture}
\draw[->] (-2,0) -- (2,0) node[right] {$x_1$}; \draw[->] (0,-2) --
(0,2) node[above] {$x_2$};
\draw[scale=0.5,domain=-3:3,smooth,variable=\x,black] plot
({\x},{\x}); \draw[scale=0.5,domain=-3:3,smooth,variable=\x,black]
plot ({\x},{-\x});
\draw[scale=0.5,domain=-3.5:2.5,smooth,variable=\x,black] plot
({\x},{-1-\x});
\node at (1.5,1.8) {$\VV(x_1-x_2)$};
\node at (-1.5,1.8) {$\VV(x_1+x_2)$};
\node at (1.5,-2) {$\VV(x_1+x_2+1)$};
\end{tikzpicture}
\caption{The variety
$\VV((x_1^2-x_2^2)(x_1+x_2+1))=\VV(x_1-x_2)\cup
\VV(x_1+x_2)\cup\VV(x_1+x_2+1)$ of $\cS$ from
Example~\ref{ex:monomialparam}. $\cS$ does not have toric steady
states according to \cite[Definition~2.2]{millan2012chemical} because
$I_2$ is not binomial. Still, $\VV(I_2)\cap \RR^n_{>0}$ is toric and
parametrized by $s \mapsto (s,s)$, for $s\in\RR_{>0}$. }
\label{fig:binomvarietyB}
\end{subfigure}
\caption{
  The positive steady state varieties of $\cT$ and $\cS$
  are equal. $\cT$ has a binomial steady state ideal while $\cS$ has 
  does not.  In both cases the equations that describe only the
  positive steady states are binomial.
}
\label{fig:binomvariety}
\end{figure}
\end{example}

The following proposition uses \cite[Section~2]{ES96} to show why the
name \emph{toric steady states} is justified.  We include it, as it
seems to have never appeared explicitly in the literature.
  Proposition~\ref{prop:binomcomp} below shows that if the steady
  state ideal is binomial, then the positive real part of its variety
  always admits a monomial parametrization, even if the corresponding
  variety does not.
\begin{prop}
\label{prop:binomcomp}
If $I\subseteq \RR[x]$ is a binomial ideal, then at most one of the
irreducible components of its variety intersects~$\RR^n_{>0}$.
\end{prop}
\begin{proof}
Without loss of generality, we can assume that
$I=I:(x_1 \ldots x_n)^{\infty}$ and $I = I\RR[x^\pm]\cap \RR[x]$ as
all other components are contained in coordinate hyperplanes.  By
\cite[Corollary~2.5]{ES96},
\begin{equation*}
I = I_+(\rho) =\langle x^{m_+}-\rho(m)x^{m_-}:m \in L_{\rho} \rangle
\end{equation*}
for a unique lattice $L\subset \ZZ^n$ and partial character
$\rho : L \to \RR^*$.  By \cite[Corollary~2.2]{ES96}, $I_+(\rho)$,
seen as an ideal of $\CC[x]$, is radical and it has a decomposition
into prime ideals as
\begin{equation*}
I_+(\rho)=\cap^g_{j=1} I_+(\rho_j),
\end{equation*}
where $\{\rho_1, \ldots, \rho_g\}$ is the set of extensions of $\rho$
to the saturation $\Sat(L)$ of $L$ and $g$ is the order of the group
$\Sat(L)/L$.  A variety $\VV(I_+(\rho_k))$ has positive points if and
only if $\rho_k$ takes only positive real values. Fixing
$b_1,\ldots,b_r$ to be a basis of $\Sat(L_{\rho})$, any basis
$c_1,\ldots,c_r$ of $L$ can be expressed in terms of the $b_i$ as
$c_i=\sum_j a_{ij}b_{j}$ where $A = (a_{ij}) \in\ZZ^{r\times r}$ has
determinant~$g$.  Let $\rho_k$ be any of the extensions of $\rho$;
since $\rho = \rho_k|_L$, we have
\begin{equation}
\label{eq:characters}
\rho (c_i)=\rho_k \left( \sum_j a_{ij}b_{j} \right)=\prod_{j}\rho_k(b_j)^{a_{ij}}.
\end{equation}
These equations in the unknowns $\rho_k(b_j)$ determine the extensions
of $\rho$ and thus the irreducible components of~$\VV(I)$.  If
$\rho_k(b_j)$ is not positive and real for some $k$ and $j\in[r]$,
then $\VV(I_+(\rho_k))\cap \RR^n_{>0} = \emptyset$.  We only need to
consider components for which $\rho_k(b_j)>0$ for all $j\in[r]$.  In
this case we can take logarithms on both sides
of~\eqref{eq:characters}:
\begin{equation}\label{eq:logeq}
\log (\rho (c_i)) = \sum_j a_{ij} \log (\rho_k(b_j)).
\end{equation}
The result is a linear equation for $\log(\rho_k(b_j))$ whose
solutions yield characters $\rho_k$ such that $\VV(I_+(\rho_k))$ has
positive points.  The matrix $A$ can be inverted over~$\QQ$.  Write
$\log \rho_k(b) = (\log \rho_k(b_1), \dots, \log \rho_k(b_r))$ and
similarly $\log\rho(c) = (\log \rho(c_1), \ldots, \log \rho (c_r))$.
Then \eqref{eq:logeq} has the unique solution
$\log \rho_k(b) = A^{-1}\log\rho(c)$. Consequently, there is a unique
saturation $\rho^* : \Sat(L) \to \RR^*$ of $\rho$ such that
$\rho^*(b_i) > 0$.
\end{proof}

With Example~\ref{ex:monomialparam} and
Proposition~\ref{prop:binomcomp} in mind, one would like to analyze
the primary decomposition of any steady state ideal that one
encounters.  If the original steady state ideal was not binomial,
then maybe the primary decomposition reveals that at least all
components whose varieties intersect the positive orthant are
binomial.  In this case one has a monomial parametrization for each
such component.  Deciding if a non-binomial variety contains positive
real points is very hard, though.  Only in the binomial case it is
easy using the analysis of characters as in the proof of
Proposition~\ref{prop:binomcomp}.

\begin{remark}\label{rem:binomcomp}
If the steady state equations in variables $(k,x)$ are binomials in
$x$, then Proposition~\ref{prop:binomcomp} holds locally.  In
particular, for any specialization of the $k$ to positive real
numbers, one has a binomial ideal, as specialization of the $k$ could
only reduce the number of terms.  For a careful analysis of the
consequences of specialization on two different generating sets of the
same steady state ideal see~\cite[Section~2]{binomialCore}.  It
remains an interesting problem to systematically analyze primary
decompositions of steady state ideals in~$\RR(k)[x]$.
\end{remark}

\section{Discussion}
\label{sec:disc}

The results in this paper show that multistationarity of
mass-action systems is a semi-algebraic condition.  Polynomial
inequalities are used to describe where in parameter space
multistationarity can occur.  On the 2-site phosphorylation network
the result is particularly satisfying as
Corollary~\ref{cor:MM-setting} complements the results of Bihan et
al.\ and shows---in biologically meaningful terms---exactly where
multistationarity is possible.  The chamber decomposition is an
interesting structure because it is inherent to the biological system.
It would be interesting to apply our methods to other systems where
the chamber decomposition is explicitly known, e.g.\ the Wnt pathway
from~\cite{Gross2016}.  For us, this shows that it is worthwhile for
biologists to interact with real algebraic geometry.  Mass-action
systems whose steady state varieties admit a monomial parametrization
appear as a natural hunting ground.  Here the techniques of this paper
can be applied and combined with ever more powerful exact
computational methods from logic.  As an immediate goal, it would be
very interesting to prove or disprove Corollary~\ref{cor:MM-setting}
for $3$-site or $n$-site phosphorylation.

\section*{Acknowledgement}
This project is funded by the Deutsche Forschungsgemeinschaft,
284057449.  Alexandru Iosif and Thomas Kahle are also partially
supported by the DFG-RTG ``MathCore'', 314838170.
We thank the anonymous reviewers for their valuable
comments. One reviewer helped to improve the paper by providing a
simpler proof of Lemma~\ref{lem:rep_multi}, clarifying the statements
of Theorems~\ref{thm:multi_c_space} and~\ref{thm:scaling_c}, and
pointing us to \cite[Theorem~4.1]{bihan2018lower} which yields
Corollary~\ref{cor:MM-setting}.

\bibliographystyle{amsplain}
\bibliography{dualsite}

\providecommand{\bysame}{\leavevmode\hbox to3em{\hrulefill}\thinspace}
\providecommand{\MR}{\relax\ifhmode\unskip\space\fi MR }
\providecommand{\MRhref}[2]{%
  \href{http://www.ams.org/mathscinet-getitem?mr=#1}{#2}
}
\providecommand{\href}[2]{#2}
\begin{thebibliography}{10}

\bibitem{inj-003}
Murad Banaji and Gheorghe Craciun, \emph{Graph-theoretic approaches to
  injectivity and multiple equilibria in systems of interacting elements},
  Communications in Mathematical Sciences \textbf{7} (2009), no.~4, 867--900.

\bibitem{inj-002}
\bysame, \emph{Graph-theoretic criteria for injectivity and unique equilibria
  in general chemical reaction systems}, Advances in Applied Mathematics
  \textbf{44} (2010), no.~2, 168 -- 184.

\bibitem{bertini}
Daniel~J. Bates, Jonathan~D. Hauenstein, Andrew~J. Sommese, and Charles~W.
  Wampler, \emph{Bertini: {S}oftware for {N}umerical {A}lgebraic {G}eometry},
  Available at bertini.nd.edu with permanent doi: dx.doi.org/10.7274/R0H41PB5.

\bibitem{becker1993computation}
Eberhard Becker and Rolf Neuhaus, \emph{Computation of real radicals of
  polynomial ideals}, Computational algebraic geometry, Springer, 1993,
  pp.~1--20.

\bibitem{bihan2018lower}
Fr{\'e}d{\'e}ric Bihan, Alicia Dickenstein, and Magal{\'\i} Giaroli,
  \emph{Lower bounds for positive roots and regions of multistationarity in
  chemical reaction networks}, preprint, arXiv:1807.05157 (2018).

\bibitem{BDEEGGHKRSW2017}
Russell Bradford, James Davenport, Matthew England, Hassan Errami, Vladimir~P.
  Gerdt, Dima Grigoriev, Charles Hoyt, Marek Kosta, Ovidiu Radulescu, Thomas
  Sturm, and Andreas Weber, \emph{A case study on the parametric occurrence of
  multiple steady states}, Proceedings of the 42nd International Symposium on
  Symbolic and Algebraic Computation (ISSAC '17), ACM, July 2017, pp.~45--52.

\bibitem{paramotopy}
Danielle Brake and Matt Niemberg, \emph{Paramotopy}, available at
  \url{http://paramotopy.com}, 2016.

\bibitem{brown2003qepcad}
Christopher~W. Brown, \emph{{QEPCAD B}: a program for computing with
  semi-algebraic sets using {CAD}s}, ACM SIGSAM Bulletin \textbf{37} (2003),
  no.~4, 97--108.

\bibitem{brown2010black}
Christopher~W. Brown and Adam Strzebo{\'n}ski, \emph{Black-box/white-box
  simplification and applications to quantifier elimination}, Proceedings of
  the 2010 International Symposium on Symbolic and Algebraic Computation, ACM,
  2010, pp.~69--76.

\bibitem{degree_paper}
Carsten Conradi, Elisenda Feliu, Maya Mincheva, and Carsten Wiuf,
  \emph{Identifying parameter regions for multistationarity}, PLOS
  Computational Biology \textbf{13} (2017), no.~10, 1--25.

\bibitem{conradi2012multistationarity}
Carsten Conradi and Dietrich Flockerzi, \emph{Multistationarity in mass action
  networks with applications to {ERK} activation}, Journal of mathematical
  biology \textbf{65} (2012), no.~1, 107--156.

\bibitem{conradi2008multistationarity}
Carsten Conradi, Dietrich Flockerzi, and J{\"o}rg Raisch,
  \emph{Multistationarity in the activation of a {MAPK}: parametrizing the
  relevant region in parameter space}, Mathematical biosciences \textbf{211}
  (2008), no.~1, 105--131.

\bibitem{conradi2014catalytic}
Carsten Conradi and Maya Mincheva, \emph{Catalytic constants enable the
  emergence of bistability in dual phosphorylation}, Journal of The Royal
  Society Interface \textbf{11} (2014), no.~95, 20140158.

\bibitem{CRNT-chapter}
Carsten Conradi and Casian Pantea, \emph{Chapter 9 - multistationarity in
  biochemical networks: Results, analysis, and examples}, Algebraic and
  Combinatorial Computational Biology (Raina Robeva and Matthew Macauley,
  eds.), Academic Press, 2019, pp.~279 -- 317.

\bibitem{fein-012}
Carsten Conradi, Julio Saez-Rodriguez, Ernst-Dieter Gilles, and J{\"o}rg
  Raisch, \emph{Using {C}hemical {R}eaction {N}etwork {T}heory to discard a
  kinetic mechanism hypothesis}, Systems Biology, IEE Proceedings (now IET
  Systems Biology) \textbf{152} (2005), no.~4, 243--248.

\bibitem{ptm-028}
Carsten Conradi and Anne Shiu, \emph{Dynamics of posttranslational modification
  systems: Recent progress and future directions}, Biophysical Journal
  \textbf{114} (2018), no.~3, 507 -- 515.

\bibitem{coste2002introduction}
Michel Coste, \emph{An introduction to semialgebraic geometry}, RAAG network
  school \textbf{145} (2002), 30.

\bibitem{cox96:_ideal_variet_algor}
David~A. Cox, John~B. Little, and Donal O'Shea, \emph{Ideals, varieties, and
  algorithms}, second ed., UTM, Springer, New York, 1996.

\bibitem{alg-017}
Gheorghe Craciun, Alicia Dickenstein, Anne Shiu, and Bernd Sturmfels,
  \emph{Toric dynamical systems}, Journal of Symbolic Computation \textbf{44}
  (2009), no.~11, 1551 -- 1565, Ordner: Gatermann.

\bibitem{craciun2009algebraic}
Gheorghe Craciun, Casian Pantea, and Grzegorz~A. Rempala, \emph{Algebraic
  methods for inferring biochemical networks: a maximum likelihood approach},
  Computational biology and chemistry \textbf{33} (2009), no.~5, 361--367.

\bibitem{de2009graphs}
Jes{\'u}s~A. De~Loera, Edward~D. Kim, Shmuel Onn, and Francisco Santos,
  \emph{Graphs of transportation polytopes}, Journal of Combinatorial Theory,
  Series A \textbf{116} (2009), no.~8, 1306--1325.

\bibitem{alg-041}
Alicia Dickenstein, \emph{Biochemical reaction networks: {A}n invitation for
  algebraic geometers}, Mathematical Congress of the Americas, vol. 656,
  American Mathematical Soc., 2016, pp.~65--83.

\bibitem{dickenstein2018multistationarity}
Alicia Dickenstein, Mercedes~Perez Millan, Anne Shiu, and Xiaoxian Tang,
  \emph{Multistationarity in structured reaction networks}, arXiv preprint
  arXiv:1810.05574 (2018).

\bibitem{dolzmann1997redlog}
Andreas Dolzmann and Thomas Sturm, \emph{{REDLOG}: Computer algebra meets
  computer logic}, SIGSAM Bull. \textbf{31} (1997), no.~2, 2--9.

\bibitem{ES96}
David Eisenbud and Bernd Sturmfels, \emph{Binomial ideals}, Duke Mathematical
  Journal \textbf{84} (1996), no.~1, 1--45.

\bibitem{fein-013}
Phillip Ellison and Martin Feinberg, \emph{How catalytic mechanisms reveal
  themselves in multiple steady-state data: {I}. {B}asic principles}, Journal
  of Molecular Catalysis A: Chemical \textbf{154} (2000), no.~1--2, 155--167.

\bibitem{fein-014}
Phillipp Ellison, Martin Feinberg, Ming-Huei Yueb, and Howard Saltsburg,
  \emph{How catalytic mechanisms reveal themselves in multiple steady-state
  data: {II}. {A}n ethylene hydrogenation example}, Journal of Molecular
  Catalysis A: Chemical \textbf{154} (2000), no.~1--2, 169--184.

\bibitem{fein-015}
Phillipp~Raymond Ellison, \emph{The {A}dvanced {D}eficiency {A}lgorithm and its
  applications to mechanism discrimination}, Ph.D. thesis, The University of
  Rochester, 1998.

\bibitem{erdi1989mathematical}
P{\'e}ter {\'E}rdi and J{\'a}nos T{\'o}th, \emph{Mathematical models of
  chemical reactions: theory and applications of deterministic and stochastic
  models}, Manchester University Press, 1989.

\bibitem{fein-016}
Martin Feinberg, \emph{The existence and uniqueness of steady states for a
  class of chemical reaction networks}, Archive for Rational Mechanics and
  Analysis \textbf{132} (1995), no.~4, 311--370.

\bibitem{fein-017}
\bysame, \emph{Multiple steady states for chemical reaction networks of
  {D}eficiency {O}ne}, Archive for Rational Mechanics and Analysis \textbf{132}
  (1995), no.~4, 371--406.

\bibitem{inj-007}
Elisenda Feliu and Carsten Wiuf, \emph{Preclusion of switch behavior in
  networks with mass-action kinetics}, Applied Mathematics and Computation
  \textbf{219} (2012), no.~4, 1449 -- 1467.

\bibitem{BMBN-site}
Dietrich Flockerzi, Katharina Holstein, and Carsten Conradi, \emph{N-site
  phosphorylation systems with 2n-1 steady states}, Bulletin of Mathematical
  Biology (2014), 1--25.

\bibitem{GawrilowJoswigPolymake}
Ewgenij Gawrilow and Michael Joswig, \emph{\texttt{polymake}: a framework for
  analyzing convex polytopes}, Polytopes---combinatorics and computation,
  vol.~29, Birkh{\"a}user, 2000, pp.~43--47.

\bibitem{GleixnerEtal2018OO}
Ambros Gleixner, Michael Bastubbe, Leon Eifler, Tristan Gally, Gerald Gamrath,
  Robert~Lion Gottwald, Gregor Hendel, Christopher Hojny, Thorsten Koch,
  Marco~E. L{\"u}bbecke, Stephen~J. Maher, Matthias Miltenberger, Benjamin
  M{\"u}ller, Marc~E. Pfetsch, Christian Puchert, Daniel Rehfeldt, Franziska
  Schl{\"o}sser, Christoph Schubert, Felipe Serrano, Yuji Shinano, Jan~Merlin
  Viernickel, Matthias Walter, Fabian Wegscheider, Jonas~T. Witt, and Jakob
  Witzig, \emph{{The SCIP Optimization Suite 6.0}}, Technical report,
  Optimization Online, July 2018.

\bibitem{Gross2016}
Elizabeth Gross, Heather~A. Harrington, Zvi Rosen, and Bernd Sturmfels,
  \emph{Algebraic systems biology: A case study for the wnt pathway}, Bulletin
  of Mathematical Biology \textbf{78} (2016), no.~1, 21--51.

\bibitem{holstein2013multistationarity}
Katharina Holstein, Dietrich Flockerzi, and Carsten Conradi,
  \emph{Multistationarity in sequential distributed multisite phosphorylation
  networks}, Bulletin of Mathematical Biology \textbf{75} (2013), no.~11,
  2028--2058.

\bibitem{Mathematica}
Wolfram~Research{,} Inc., \emph{Mathematica, {V}ersion 11.2}, Champaign, IL,
  2017.

\bibitem{kahle11mesoprimary}
Thomas Kahle and Ezra Miller, \emph{Decompositions of commutative monoid
  congruences and binomial ideals}, Algebra \& Number Theory \textbf{8} (2014),
  no.~6, 1297--1364.

\bibitem{LAZARD2007636}
Daniel Lazard and Fabrice Rouillier, \emph{Solving parametric polynomial
  systems}, Journal of Symbolic Computation \textbf{42} (2007), no.~6, 636 --
  667.

\bibitem{maple}
{Maplesoft, a division of Waterloo Maple Inc., Waterloo, Ontario}, \emph{Maple
  2017}.

\bibitem{sig-005}
Nick~I. Markevich, Jan~B. Hoek, and Boris~N. Kholodenko, \emph{Signaling
  switches and bistability arising from multisite phosphorylation in protein
  kinase cascades}, The Journal of Cell Biology \textbf{164} (2004), no.~3,
  353--359.

\bibitem{millan2012chemical}
Mercedes~P{\'e}rez Mill{\'a}n, Alicia Dickenstein, Anne Shiu, and Carsten
  Conradi, \emph{Chemical reaction systems with toric steady states}, Bulletin
  of Mathematical Biology \textbf{74} (2012), no.~5, 1027--1065.

\bibitem{inj-006}
Stefan M{\"u}ller, Elisenda Feliu, Georg Regensburger, Carsten Conradi, Anne
  Shiu, and Alicia Dickenstein, \emph{Sign conditions for injectivity of
  generalized polynomial maps with applications to chemical reaction networks
  and real algebraic geometry}, Foundations of Computational Mathematics
  \textbf{16} (2016), no.~1, 69--97.

\bibitem{neuhaus1998computation}
Rolf Neuhaus, \emph{Computation of real radicals of polynomial ideals—{II}},
  Journal of Pure and Applied Algebra \textbf{124} (1998), no.~1-3, 261--280.

\bibitem{MESSI}
Mercedes Pérez~Millán and Alicia Dickenstein, \emph{The structure of {MESSI}
  biological systems}, SIAM Journal on Applied Dynamical Systems \textbf{17}
  (2018), no.~2, 1650--1682.

\bibitem{Rambau:TOPCOM-ICMS:2002}
J{\"o}rg Rambau, \emph{{TOPCOM}: Triangulations of point configurations and
  oriented matroids}, Mathematical Software---ICMS 2002 (Arjeh~M. Cohen,
  Xiao-Shan Gao, and Nobuki Takayama, eds.), World Scientific, 2002,
  pp.~330--340.

\bibitem{lin-004}
Ralph~Tyrell Rockafellar, \emph{Convex analysis}, Princeton University Press,
  1970.

\bibitem{sadeghimanesh2018groebner}
AmirHosein Sadeghimanesh and Elisenda Feliu, \emph{Groebner bases of reaction
  networks with intermediate species}, preprint, arXiv:1804.01381 (2018).

\bibitem{binomialCore}
\bysame, \emph{The multistationarity structure of networks with intermediates
  and a binomial core network}, preprint, arXiv:1808.07548 (2018).

\bibitem{fein-008}
Paul~M. Schlosser and Martin Feinberg, \emph{A theory of multiple steady states
  in isothermal homogeneous {CFSTRs} with many reactions}, Chemical Engineering
  Science \textbf{49} (1994), no.~11, 1749 -- 1767.

\bibitem{fein-045}
Guy Shinar and Martin Feinberg, \emph{Concordant chemical reaction networks},
  Mathematical Biosciences \textbf{240} (2012), no.~2, 92 -- 113.

\bibitem{fein-060}
\bysame, \emph{Concordant chemical reaction networks and the species-reaction
  graph}, Mathematical Biosciences \textbf{241} (2013), no.~1, 1 -- 23.

\bibitem{shiu2010algebraic}
Anne Shiu, \emph{Algebraic methods for biochemical reaction network theory},
  Ph.D. thesis, University of California, Berkeley, 2010.

\bibitem{alg-030}
Anne Shiu and Bernd Sturmfels, \emph{Siphons in chemical reaction networks},
  Bulletin of Mathematical Biology \textbf{72} (2010), no.~6, 1448--1463,
  Ordner: Gatermann.

\bibitem{multi-001}
Liming Wang and Eduardo Sontag, \emph{On the number of steady states in a
  multiple futile cycle}, Journal of Mathematical Biology \textbf{57} (2008),
  29--52.

\bibitem{inj-010}
Carsten Wiuf and Elisenda Feliu, \emph{Power-law kinetics and determinant
  criteria for the preclusion of multistationarity in networks of interacting
  species}, SIAM Journal on Applied Dynamical Systems \textbf{12} (2013),
  no.~4, 1685--1721.

\bibitem{ziegler2012lectures}
G{\"u}nter~M. Ziegler, \emph{Lectures on polytopes}, GTM, vol. 152, Springer
  Verlag, 2012.

\end{thebibliography}

\end{document}